\documentclass[notitlepage,a4,10.5pt]{article}

\usepackage[backend=biber,style=alphabetic,doi=true,url=false]{biblatex}
\usepackage[colorlinks=true,linkcolor=blue,citecolor=blue,urlcolor=blue]{hyperref}

\usepackage{tikz}
\usetikzlibrary{arrows.meta,positioning,calc,shapes,fit}
\usepackage{subcaption}

\usepackage{tikz-cd}

\usepackage{array}
\usepackage{booktabs}  
\usepackage{multirow}  
\usepackage[margin=1in]{geometry}

\usepackage{graphicx}

\usepackage{amsmath}%
\usepackage{amsfonts}%
\usepackage{amssymb}%
\usepackage{mathrsfs}
\usepackage{amsthm}
\usepackage{float}

\usepackage{authblk}
\usepackage{enumitem}

\usepackage{xcolor}
\usepackage{appendix}

\newcommand{\mc}{\mathcal}
\newcommand{\cp}{\times}

\newcommand{\bol}{\boldsymbol}

\newcommand{\abs}[1]{\left\lvert{#1}\right\rvert}
\newcommand{\w}{\wedge}
\newcommand{\lr}[1]{\left({#1}\right)}
\newcommand{\lrs}[1]{\left[{#1}\right]}
\newcommand{\lrc}[1]{\left\{{#1}\right\}}

\newcommand{\mf}{\mathfrak}
\newcommand{\p}{\partial}

\newcommand{\ti}[1]{\textit{#1}}

\newtheorem{remark}{\textit{Remark}}[section]

\newtheorem{proposition}{\textit{Proposition}}[section]
\newtheorem{question}{\textit{Question}}

\newcommand{\eq}[1]{\begin{equation}\begin{split}{#1}\end{split}\end{equation}}
\newcommand{\sys}[2]{\begin{subequations}\begin{align}{#1}\end{align}\label{#2}\end{subequations}}

\begin{document}

\title{
Quantum--Fluid Correspondence for Systems of Nonrelativistic Spin-$\frac{1}{2}$ Particles 
}
\author{Naoki Sato 
\thanks{Corresponding author} {}\thanks{National Institute for Fusion Science,  322-6 Oroshi-cho Toki-city, Gifu 509-5292, Japan,
Email: \href{sato.naoki@nifs.ac.jp}{sato.naoki@nifs.ac.jp}
}  
{}\thanks{Graduate School of Frontier Sciences, The University of Tokyo,  
Kashiwa, Chiba 277-8561, Japan}
\qquad 
Michio Yamada \thanks{Research Institute for Mathematical Sciences, Kyoto University, Kyoto 606-8502, Japan,\ Email: \href{yamada@kurims.kyoto-u.ac.jp}{yamada@kurims.kyoto-u.ac.jp}} 
}
\date{\today}
\setcounter{Maxaffil}{0}
\renewcommand\Affilfont{\itshape\small}

    \maketitle
    \begin{abstract}
We show that a charged fluid endowed with an internal spin degree of freedom naturally satisfies the Pauli equation for a nonrelativistic spin-\(\tfrac12\) particle, and that a collection of \(n\) such interacting fluids can be reformulated as an Euler flow in \(3n\) dimensions, thereby providing a natural representation of a system of \(n\) Pauli particles. These results provide a fluid-mechanical derivation of the Pauli equation and extend the Madelung, or quantum-hydrodynamic, picture to many-particle quantum systems. In particular, they imply that an \(n\)-qubit quantum computer can, at least in principle, be realized as a suitable combination of \(n\) fluids, or equivalently as a \(3n\)-dimensional Euler flow.
\end{abstract}

\tableofcontents

\section{Introduction}

The Madelung transform provides a hydrodynamic reformulation of quantum mechanics by expressing a wave function in terms of fluid variables such as density and velocity. 
In its original form, due to Madelung \cite{Madelung1926,Madelung1927}, the Schr\"odinger equation becomes mathematically equivalent to an irrotational compressible Euler flow subject to a quantum force. 
Analogous hydrodynamic reformulations also exist for the Pauli and Dirac equations, provided one introduces, respectively, spin degrees of freedom and both spin and phase degrees of freedom \cite{Tak54,Tak55,Tak56}. 
For modern accounts of Bohmian and quantum-hydrodynamic formulations, see also \cite{Holland93,Wyatt05}.

The aim of the present paper is to investigate, in the context of the Pauli equation, the converse point of view. 
More precisely, we ask whether a charged fluid endowed with an internal spin (intrinsic angular momentum) degree of freedom---which we call a \emph{fluid with spin}---can be shown to obey the hydrodynamic (Madelung) form of the Pauli equation, and we examine the theoretical and computational consequences of such an inverse quantum--fluid correspondence. 
Specifically, we address the following two questions:

\begin{question}\label{Q1}
Given a fluid with spin, can one derive the hydrodynamic (Madelung) form of the Pauli equation from fluid-mechanical principles?
\end{question}

\begin{question}\label{Q2}
Can \(N\) interacting fluids with spin be used to describe the dynamics of a system of \(N\) Pauli particles?
\end{question}

\noindent We will show in sections \ref{sec:PauliDer} and \ref{sec:NPauli} that the answer to both Questions \ref{Q1} and \ref{Q2} is affirmative. 
Before reviewing the background and motivation that led us to consider these problems, we summarize what we regard as the main implications of this result:

\begin{figure}
  \centerline{\includegraphics[scale=0.26]{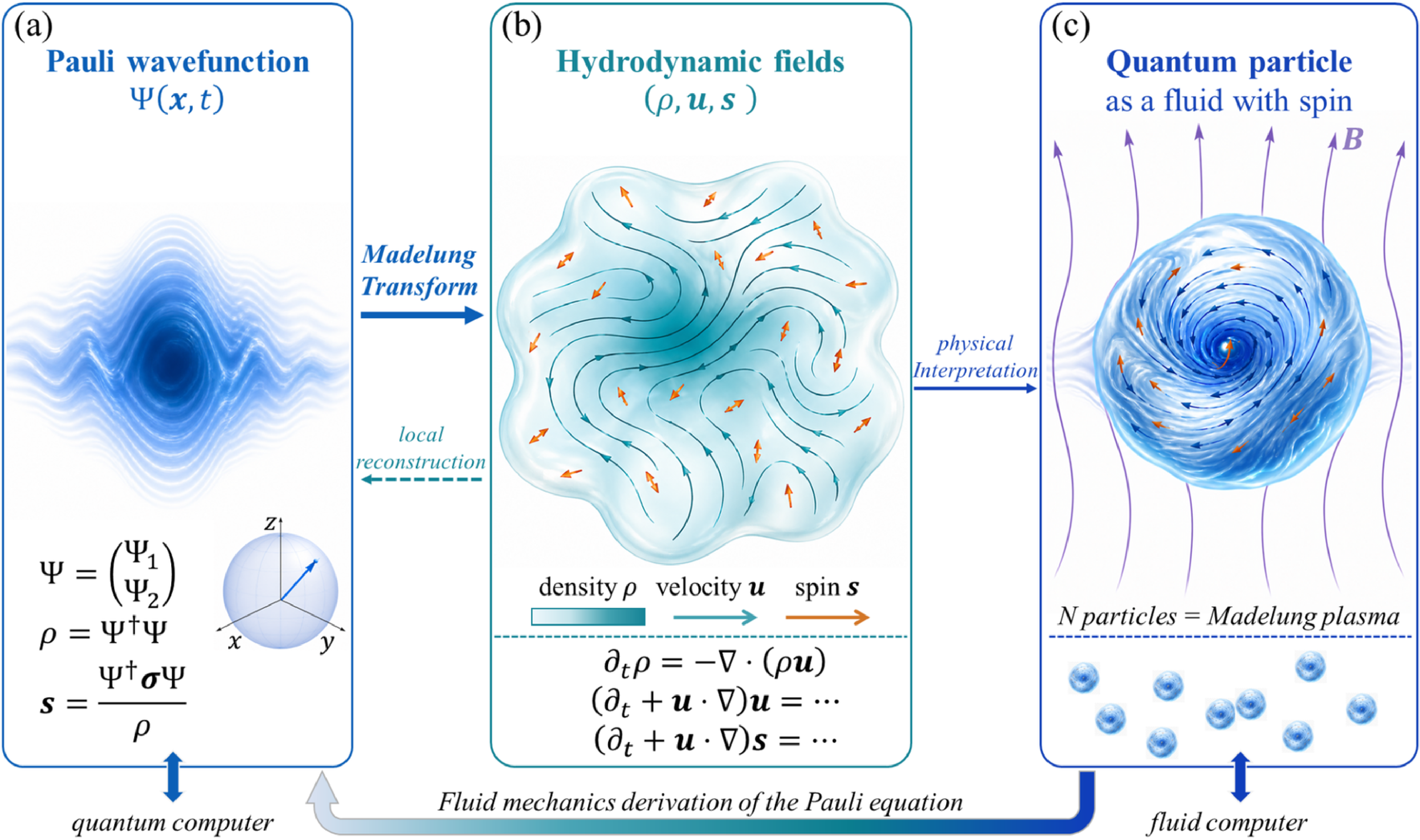}}
  \caption{
  Schematic illustration of the quantum--fluid correspondence developed in the present work. 
  Panel (a) shows the usual Pauli wavefunction description and its Madelung transform into the hydrodynamic variables \(\lr{\rho,\bol{u},\bol{s}}\), consisting of the density, velocity, and spin fields. 
  Panel (b) represents the corresponding hydrodynamic system, whose governing equations are those of the Madelung form of the Pauli equation. 
  Panel (c) summarizes the converse viewpoint explored in this paper: a quantum particle may be modeled as a fluid with spin, and a system of \(N\) Pauli particles may be represented by \(N\) interacting fluids with spin, or equivalently by a higher-dimensional Madelung plasma. 
  The lower arrows indicate the computational perspective suggested by this correspondence, namely the possibility, at least in principle, of realizing quantum dynamics through suitably engineered fluid systems, and conversely of simulating fluid-with-spin systems on quantum computers.
  }
\label{fig1}
\end{figure}

\begin{enumerate}
\item A quantum-mechanical evolution can be derived, in hydrodynamic form, from fluid-mechanical principles if one models a quantum particle as a fluid with spin. 
The resulting Madelung system is mathematically equivalent to the corresponding quantum-mechanical equation and is therefore fully consistent with the standard formulation of quantum mechanics.

\item In principle, an \(N\)-fluid system with spin can be used to represent a system of \(N\) Pauli particles, and hence an \(N\)-qubit quantum computer. 
Such a program would require the identification or implementation of a fluid bit, together with its extension to a fluid quantum bit by incorporating the quantum forces and stresses characteristic of spin dynamics in a Madelung fluid. 
Conversely, at least in principle, an \(N\)-fluid system with spin could itself be simulated on a quantum computer.
\end{enumerate}

\noindent Figure \ref{fig1} summarizes the main ideas, scope, and motivations of the quantum--fluid correspondence examined in the present work. 
In particular, it highlights the two complementary directions developed below: first, the passage from the Pauli wavefunction description to its hydrodynamic formulation through the Madelung transform; and second, the converse derivation showing that a fluid endowed with an internal spin degree of freedom obeys the same hydrodynamic system. 
It also illustrates the many-particle extension discussed later in the paper, together with the computational perspective motivating the notion of a fluid realization of quantum dynamics.

We now review the background and scientific literature surrounding the quantum--fluid correspondence suggested by the Madelung transform.

From a physical point of view, the Madelung transform offers a hydrodynamic perspective on quantum phenomena, opening the possibility of applying fluid-mechanical intuition and techniques to the study of quantum dynamics. 
Historically, pilot-wave theory \cite{Bohm52a,Bohm52b}, in which the Schr\"odinger particle velocity is identified with the hydrodynamic velocity field obtained from the Madelung transform, provides a well-known example of how the fluid picture can lead to new viewpoints on the interpretation of quantum mechanics and on the scope and limitations of hidden-variable theories; see, for example, the discussion in \cite{Tak52}, as well as \cite{EinsteinPodolskyRosen35} and the surrounding literature on local realism. 
For the Pauli equation in the causal/hydrodynamic setting, see also \cite{Bohm55}. 
We stress, however, that the present work does not rely on any particular interpretation of quantum mechanics. 
Whether one works in the standard spinor formulation or in its Madelung form, the governing equations have the same mathematical content and admit the same empirical interpretation. 
Our concern here is instead with the structural consequences of modeling a quantum system as a fluid with spin, specifically in connection with Questions \ref{Q1} and \ref{Q2}.

The possible relevance of a fluid model for the spin-dependent terms---that is, for the terms involving the reduced Planck constant \(\hbar\) that encode the characteristic nonclassical features of quantum particles---has already been explored in several works. 
The basic idea is that the quantum potential appearing in the Madelung form of the Schr\"odinger equation may be interpreted as the kinetic energy of an internal motion, often associated with the so-called Zitterbewegung \cite{Schrodinger30,Hestenes90}, which then acts effectively as a potential energy on the Madelung fluid; see \cite{Salesi96,Recami98,Esposito99} and the references therein. 
In particular, by applying Maxwell's equations to determine the magnetic field generated by the spinning charge associated with such an internal fluid motion, one can recover the corresponding velocity field, evaluate its kinetic energy, and thereby obtain the Madelung form of the Schr\"odinger equation \cite{Sato25QF}. 
Thus, within that framework, a fluid with spin provides a fluid-mechanical derivation of the Schr\"odinger equation from the motion of charged fluid parcels.

Our first motivation is to show that the same fluid-with-spin model also yields the hydrodynamic form of the Pauli equation, once higher-order corrections to the kinetic energy of the internal spin motion are retained. 
In this way, the fluid-mechanical picture is shown not to be restricted to the scalar Schr\"odinger particle, but to extend to the nonrelativistic dynamics of particles carrying spin.

A second motivation comes from the observation that the Madelung form of the Pauli equation, and more generally the hydrodynamic formulation of quantum systems, also suggests a natural bridge between quantum mechanics and general relativity. 
Indeed, fluid systems are intrinsically compatible with general relativity, since they may be described by inserting a fluid stress-energy-momentum tensor into the Einstein field equations; see, for example, \cite{Frankel13,MTW17} on the perfect-fluid stress-energy tensor, and the classic papers of Tolman \cite{Tol30a,Tol30b}, where related ideas are used in the study of relativistic temperature. 
Starting from the hydrodynamic form of a quantum equation, one may derive the corresponding Madelung-fluid stress-energy-momentum tensor and use it as a source term in the Einstein equations, thereby suggesting an alternative route to the coupling between quantum dynamics and gravitation. 
This point of view, together with its scope and limitations, is discussed in \cite{Sato25QF}.

At the many-particle level, Bohmian and quantum-hydrodynamic formulations have also been developed in \cite{Renziehausen18,Renziehausen20}, where the relation between many-particle Bohmian mechanics and hydrodynamic balance laws is analyzed in detail. 
The viewpoint of the present work is different: rather than starting from a many-particle Schr\"odinger equation and projecting it to hydrodynamic variables, we ask whether interacting fluids with spin can themselves be used, in the converse direction, to represent a system of \(N\) Pauli particles.

From a mathematical point of view, the possibility of encoding quantum-mechanical equations into the dynamics of a fluid system may also be viewed as an instance of the \ti{universality} of the Euler equations governing ideal fluid flows, namely their ability to represent a broad class of differential equations after suitably enlarging the space dimension and choosing an appropriate ambient metric. 
This universality was developed by Tao \cite{Tao18,Tao20}, who showed that all quadratic ODEs satisfying suitable symmetry and conservation properties can be embedded into the Euler equations on a suitable Riemannian manifold. 
Later, it was shown in \cite{Cardona23} that non-autonomous dynamical systems can be extended to smooth steady Beltrami flows on suitable Riemannian manifolds.

Such universal features of the Euler equations naturally raise the question of whether fluid systems may be used to perform computation \cite{Moore90}, and more specifically whether suitably designed flows can realize universal Turing machines. 
It is useful to distinguish several conceptually different notions of ``fluid computing.'' 
At a first level, represented by the liquid-computing paradigm surveyed by Adamatzky \cite{Adamatzky19} and by recent reviews of fluidic logic for biohybrid computation \cite{Singh23}, the fluid serves primarily as a signal carrier or as a reactive medium, and computation is realized through explicit logical operations involving pressure levels, fluid jets, droplets, liquid marbles, or chemical wave fronts. 
At a second, more intrinsically dynamical, level, Cardona and collaborators \cite{Cardona21fc,Cardona25fc} showed that a steady Euler flow may itself realize universal computation, not through an externally prescribed gate architecture, but through the symbolic dynamics of its Lagrangian trajectories, leading to Turing-complete stationary Euler flows with undecidable orbit properties. 
The viewpoint adopted in the present work is different from both of these. 
Here the elementary computational unit is not a localized signal or gate, nor a single distinguished particle trajectory, but an entire Madelung fluid, regarded as a fluid-quantum-bit whose information is encoded in its collective hydrodynamic state.

The Madelung transform thus provides a concrete example in which the solution of a physical problem---namely, the computation of the dynamical state of a quantum particle---can be reformulated as the solution of a hydrodynamic system. 
Conversely, a suitably designed quantum computer could in principle be used to solve the evolution of a fluid system. 
Indeed, it has been suggested that fluid dynamics may provide a natural setting in which quantum advantage emerges, that is, a setting in which quantum computers achieve a computational speed-up over their classical counterparts; see, for example, \cite{Vidal03,Xu25,Cicero25} on the classical simulation of quantum computers, and \cite{Meng23,Meng24a,Meng24b} and the references therein on the quantum solution of fluid equations, including the Navier--Stokes equations.

Answering Question \ref{Q2} allows one to consider the converse perspective, namely the possibility of realizing or simulating an \(N\)-qubit quantum computer by means of an \(N\)-fluid system. 
It is not, however, the aim of the present paper to propose a practical architecture, to estimate computational speed-ups, or to assess implementation issues for such a fluid quantum computer. 
Our purpose is more modest and more structural: to show that the quantum--fluid correspondence developed here makes such a viewpoint mathematically meaningful.

The present manuscript is organized as follows. 
In section \ref{sec:Mad}, we review the Madelung transform for two-component spinor fields and discuss the inverse reconstruction problem. 
In section \ref{sec:Pauli}, we derive the hydrodynamic form of the Pauli equation and show that the Madelung correspondence propagates along the flow. 
In section \ref{sec:PauliDer}, we show that the same hydrodynamic system arises from a charged fluid endowed with an internal spin degree of freedom. 
In section \ref{sec:Hamiltonian}, we discuss the associated Hamiltonian structure. 
Finally, in section \ref{sec:NPauli}, we extend the construction to interacting many-particle systems and show that they can be reformulated as higher-dimensional fluids with spin.

\subsection{Acknowledgments}
NS acknowledges helpful discussions with P. J. Morrison during his visit to UT Austin, and with K. Seto.  
The research of NS was partially supported by JSPS KAKENHI Grant
No. 25K07267, 
No.  22H04936, and No. 24K00615.

\section{The Madelung transform of spinor fields}\label{sec:Mad}

The aim of this section is to describe the Madelung transform,
\eq{
\mc{M}:\Psi\mapsto\lr{\rho,\bol{u},\bol{s}},\qquad \mc{M}\lr{\Psi}=\lr{\rho,\bol{u},\bol{s}},\label{MT}
}
which maps a two-component spinor field $\Psi$ to a set of hydrodynamic variables, namely the density, velocity, and spin fields \(\lr{\rho,\bol{u},\bol{s}}\), and to discuss its inverse. 

Let
\eq{
\Psi=\Psi\lr{\bol{x},t}:\Omega\times[0,+\infty)\rightarrow\mathbb{C}^2,
}
where \(\Omega\subseteq\mathbb{R}^3\), denote a two-component spinor field. Introducing the reduced Planck constant $\hbar$, we write
\begin{equation}\label{spinor}
\Psi=\begin{pmatrix}\Psi_1\\\Psi_2\end{pmatrix}
=\begin{pmatrix}\sqrt{\rho_1} \exp\lrc{{\rm i}\frac{\theta_1}{\hbar}}\\\sqrt{\rho_2}\exp\lrc{{\rm i}\frac{\theta_2}{\hbar}}\end{pmatrix},
\end{equation}
for real valued densities $\rho_i\lr{\bol{x},t}\geq 0$ and phases $\theta_i\lr{\bol{x},t}$, $i=1,2$.

In order to obtain the hydrodynamic 
(Madelung) form of  the Pauli equation (see eq.~\eqref{Pauli} below), we need to introduce 
real valued hydrodynamic fields that encode the governing equation in fluid form. 
First, the density $\rho\lr{\bol{x},t}$ is defined as
\begin{equation}
\rho=\Psi^\dagger \Psi=
\begin{pmatrix}
\Psi_1^\ast,\Psi_2^\ast
\end{pmatrix}
\begin{pmatrix}
\Psi_1\\\Psi_2
\end{pmatrix}=\abs{\Psi_1}^2+\abs{\Psi_2}^2=\rho_1+\rho_2.\label{rho}
\end{equation}
Next, 
the velocity field $\bol{u}\lr{\bol{x},t}$ is given by
\eq{
\bol{u}=\frac{\hbar}{2m{\rm i}\rho}\lr{\Psi^\dagger\nabla\Psi-\nabla\Psi^{\dagger} \Psi}-\frac{q}{m}\bol{A}=\frac{1}{m\rho}\lr{\rho_1\nabla\theta_1+\rho_2\nabla\theta_2}-\frac{q}{m}\bol{A},\label{u}
}
where \(m\) and \(q\) denote the particle mass and charge, respectively, and $\bol{A}\lr{\bol{x},t}$ is the magnetic vector potential. 
Let \(\boldsymbol{\sigma}=\lr{\sigma_1,\sigma_2,\sigma_3}\) denote the Pauli matrices,
\begin{equation}
\sigma_1=
\begin{pmatrix}
0&1\\
1&0
\end{pmatrix},
\qquad
\sigma_2=
\begin{pmatrix}
0&-i\\
i&0
\end{pmatrix},
\qquad
\sigma_3=
\begin{pmatrix}
1&0\\
0&-1
\end{pmatrix}. 
\end{equation}
We introduce the relative phase
\begin{equation}
\varphi=\frac{\theta_2-\theta_1}{\hbar},
\end{equation}
and the angle $\eta$ by
\begin{equation}
\rho_1=\rho \cos^2\!\lr{\frac{\eta}{2}},
\qquad
\rho_2=\rho \sin^2\!\lr{\frac{\eta}{2}}.\label{eta}
\end{equation} 
The (normalized) spin field is then defined by 
\begin{equation}\label{s}
\bol{s}=\frac{\Psi^\dagger\bol{\sigma}\Psi}{\Psi^\dagger\Psi}=\frac{1}{\rho}\begin{pmatrix}
\Psi_1^\ast\Psi_2+\Psi_2^\ast\Psi_1\\{\rm i}\Psi_2^\ast\Psi_1-{\rm i}\Psi_1^\ast\Psi_2\\
\Psi_1^\ast\Psi_1-\Psi_2^\ast\Psi_2
\end{pmatrix}=
\begin{pmatrix}
\sin\eta\cos\varphi\\
\sin\eta\sin\varphi\\\cos\eta
\end{pmatrix}.  
\end{equation} 

Now observe that, for a given spinor field \(\Psi\) and vector potential \(\bol{A}\), the hydrodynamic variables \(\lr{\rho,\bol{u},\bol{s}}\) are uniquely determined by \eqref{rho}, \eqref{u}, and \eqref{s}. A natural question is whether the converse also holds, that is, whether a given triple of hydrodynamic fields \(\lr{\rho,\bol{u},\bol{s}}\) determines a unique global spinor field \(\Psi\).
In the case of the Madelung transform for the Schr\"odinger equation \cite{Madelung1926,Madelung1927}, the complex wave function
\eq{
\psi=\sqrt{\rho}\exp\lrc{{\rm i}\frac{\theta}{\hbar}}\label{wf},
}
determines a pair of hydrodynamic variables \(\lr{\rho,\bol{u}}\), where \(m\bol{u}=\nabla\theta\). Conversely, if \(\Omega\subseteq\mathbb{R}^3\) has trivial first de Rham cohomology and if a pair \(\lr{\rho,\bol{u}}\) satisfies \(\nabla\cp\bol{u}=\bol{0}\), then there exists a globally defined phase \(\theta\), unique up to an additive constant, such that $m\bol{u}=\nabla\theta$,
and hence one recovers a unique wave function \(\psi=\sqrt{\rho}\exp\lrc{{\rm i}\theta/\hbar}\) up to multiplication by a constant phase factor.  
However, the hydrodynamic pair \(\lr{\rho,\bol{u}}\) is generally not sufficient to reconstruct a unique global spinor \(\Psi\) in the Pauli case. Proposition~\ref{prop:spinor}, together with Remarks~\ref{rem:sch},  \ref{rem:unique}, and \ref{rem:local} clarifies this point: away from vacuum points and under a nondegeneracy assumption on the canonical vorticity, a spinor field \(\Psi\) may be reconstructed locally from hydrodynamic data \(\lr{\rho,\bol{u}}\). In general, this reconstruction is only local and need not be unique; the nonuniqueness reflects the gauge freedom inherent in Clebsch representations of differential \(1\)-forms \cite[sec.~IID]{Yoshida2009Clebsch}.

\begin{proposition}[Local reconstruction of a spinor field from hydrodynamic data]
\label{prop:spinor}
Let \(\rho\) and \(\bol{u}\) be smooth density and velocity fields on a domain
\(\Omega\subseteq\mathbb{R}^3\), and let \(\bol{A}\) be a smooth vector potential on \(\Omega\).
Fix \(\bol{x}_0\in\Omega\), and assume that
\eq{
\rho(\bol{x}_0)>0,
\qquad
\nabla\cp\bol{p}(\bol{x}_0)\neq \bol{0},
}
where the canonical momentum is defined by
\eq{
\bol{p}=m\bol{u}+q\bol{A}.
}
Then there exists a neighborhood \(U\subset\Omega\) of \(\bol{x}_0\), together with smooth functions 
$\rho_1,\rho_2,\theta_1,\theta_2:U\to\mathbb{R}$, 
such that
\eq{
\rho_1+\rho_2=\rho,
\qquad
\rho_1>0,
\qquad
\rho_2>0
\quad{\rm in}~~U,
}
and
\eq{
\frac{\rho_1}{\rho}\nabla(\theta_1-\theta_2)+\nabla\theta_2
=
m\bol{u}+q\bol{A}
\quad{\rm in}~~U.
}
Consequently, the spinor field
\begin{equation}
\Psi=
\begin{pmatrix}
\Psi_1\\
\Psi_2
\end{pmatrix},
\qquad
\Psi_i=\sqrt{\rho_i}\exp\lrc{{\rm i}\frac{\theta_i}{\hbar}},
\quad i=1,2,
\end{equation}
is smooth in \(U\).
\end{proposition}

\begin{proof}
Let
$p=\bol{p}^{\flat}$ 
denote the \(1\)-form associated with the canonical momentum \(\bol{p}=m\bol{u}+q\bol{A}\), and let
$w=dp$ 
denote the associated canonical vorticity \(2\)-form. Since \(w\) is closed and \(w(\bol{x}_0)\neq 0\), after restricting to a sufficiently small neighborhood \(V\subset\Omega\) of \(\bol{x}_0\), we may assume that \(w\neq 0\) throughout \(V\). Hence \(w\) has constant rank \(2\) in \(V\), and by the Lie--Darboux theorem \cite{Gracia24}[thm 2.10] there exist smooth functions \(\alpha,\beta\in C^\infty(V)\) such that
\eq{
w=d\alpha\w d\beta.
}
It follows that
\eq{
d\lr{p-\alpha\,d\beta}=dp-d\alpha\w d\beta=0.
}
After shrinking \(V\) if necessary, we may assume that \(V\) is contractible. The Poincar\'e lemma then yields a smooth function \(\gamma\in C^\infty(V)\) such that
\eq{
p=\alpha\,d\beta+d\gamma.
} 
We next use the elementary Clebsch gauge freedom
\eq{
\alpha\,d\beta+d\gamma
=
(\alpha+C)\,d\beta+d(\gamma-C\beta),
\qquad C\in\mathbb{R},
}
to arrange that \(\alpha>0\) in \(V\). Indeed, after possibly shrinking \(V\), the function \(\alpha\) is bounded below in \(V\), and one may choose \(C\) so that \(\alpha+C>0\) throughout \(V\). Replacing \((\alpha,\beta,\gamma)\) by
$(\alpha+C,\beta,\gamma-C\beta)$,
we may therefore assume, without loss of generality, that
$
\alpha>0$ in $V$. 
Since \(\rho(\bol{x}_0)>0\) and \(\rho\) is continuous, after shrinking \(V\) if necessary we may also assume that
$\rho>0$ in $V$. Choose an open set \(U\Subset V\) and a constant \(M>0\) such that
$M>\sup_U \alpha$.
Then $0<\frac{\alpha}{M}<1$ in $U$. 
Define
\eq{
\rho_1=\frac{\alpha}{M}\rho,
\qquad
\rho_2=\rho-\rho_1
=
\lr{1-\frac{\alpha}{M}}\rho,
}
and
\eq{
\theta_2=\gamma,
\qquad
\theta_1=\gamma+M\beta.
}
By construction,
\eq{
\rho_1+\rho_2=\rho,
\qquad
\rho_1>0,
\qquad
\rho_2>0
\quad{\rm in}~~U.
}
Moreover,
\eq{
\frac{\rho_1}{\rho}\,d(\theta_1-\theta_2)+d\theta_2
=
\frac{\alpha}{M}\,d(M\beta)+d\gamma
=
\alpha\,d\beta+d\gamma
=
p.
}
Applying the musical isomorphism \(\sharp\), we obtain
\eq{
\frac{\rho_1}{\rho}\nabla(\theta_1-\theta_2)+\nabla\theta_2
=
\bol{p}
=
m\bol{u}+q\bol{A}.
}
Finally, define
\eq{
\Psi_i=\sqrt{\rho_i}\exp\lrc{{\rm i}\frac{\theta_i}{\hbar}},
\qquad i=1,2.
}
Since \(\rho_i>0\) and \(\theta_i\) are smooth in \(U\), it follows that each \(\Psi_i\) is smooth in \(U\), and hence so is \(\Psi\).
\end{proof}

\begin{remark}[Vanishing canonical vorticity and reduction to a single-component wave function]\label{rem:sch}
If \(\nabla\cp\bol{p}=\bol{0}\) in some open set \(U\subset\Omega\), then, after shrinking \(U\) if necessary, the canonical momentum \(1\)-form \(p=\bol{p}^{\flat}\) is closed and hence exact (Poincar\'e lemma). Hence, there exists a smooth function \(\theta\in C^{\infty}(U)\) such that
$p=d\theta$. 
In this case, the two-component spinor field \(\Psi\) may be replaced locally by the scalar wave function
\eq{
\psi=\sqrt{\rho}\exp\lrc{{\rm i}\frac{\theta}{\hbar}}.
}
Thus, vanishing canonical vorticity corresponds locally to the standard Schr\"odinger--Madelung regime. In particular this regime does not support smooth vorticity \cite{Yoshida2016QuantumSpirals}.
\end{remark}

\begin{remark}[Nonuniqueness of the reconstructed spinor field]
\label{rem:unique}
The spinor field obtained in Proposition~\ref{prop:spinor} is generally not unique.
Indeed, although \(\Psi\) determines the hydrodynamic variables \(\lr{\rho,\bol{u}}\) uniquely through \eqref{rho} and \eqref{u} once \(\lr{\bol{A},\Phi}\) are fixed, the converse reconstruction depends on the choice of Clebsch potentials. Distinct triples \((\alpha,\beta,\gamma)\) may determine the same canonical momentum \(1\)-form \(p\), and hence the same hydrodynamic fields, while producing different phase decompositions and therefore different spinor fields. 
For example, if
\eq{
\alpha'=-\beta,
\qquad
\beta'=\alpha,
\qquad
\gamma'=\gamma+\alpha\beta,
}
then
\eq{
\alpha\,d\beta+d\gamma
=
\alpha'\,d\beta'+d\gamma'.
}
Accordingly, the spinor fields induced by \((\alpha,\beta,\gamma)\) and \((\alpha',\beta',\gamma')\) need not coincide.
Suppressing the nonuniqueness of the reconstructed spinor field requires additional physical constraints (e.g., boundary conditions), on the Clebsch potentials.  
\end{remark}

\begin{remark}[Local character of the reconstruction]\label{rem:local}
The construction in Proposition~\ref{prop:spinor} is local. A global reconstruction may fail in general, unless one allows the Clebsch potentials \(\alpha,\beta,\gamma\), and hence the phases \(\theta_i\), to be multivalued (for instance, angle-valued) functions. The existence of a global spinor field of this type depends on the topology of \(\Omega\), the boundary conditions, and the global structure of the velocity field \(\bol{u}\).
\end{remark} 
The locality and nonuniqueness of the inverse Schr\"odinger--Madelung transform 
$(\rho,\bol{u}) \mapsto \Psi$ explain why the global and unique reconstruction of a spinor field \(\Psi\) from hydrodynamic data requires additional information about the spin field \(\bol{s}\).
More precisely, in Proposition~\ref{prop:GlobalSpinor} we will show that, if the initial hydrodynamic fields
$(\rho_0,\bol{u}_0,\bol{s}_0)
=
\bigl(\rho(\bol{x},0),\bol{u}(\bol{x},0),\bol{s}(\bol{x},0)\bigr)$ 
are induced by a given initial spinor field \(\Psi_0=\Psi(\bol{x},0)\), that is,
$\mc{M}(\Psi_0)= (\rho_0,\bol{u}_0,\bol{s}_0)$,
then, under suitable existence and uniqueness assumptions for the hydrodynamic fields
$\bigl(\rho(\bol{x},t),\bol{u}(\bol{x},t),\bol{s}(\bol{x},t)\bigr)$, 
the hydrodynamic solution remains
exactly the Madelung transform of the corresponding Pauli solution
 at later times, $\mc{M}\lr{\Psi}=\lr{\rho,\bol{u},\bol{s}}$.

\section{The hydrodynamic form of the Pauli equation for a nonrelativistic spin-$\frac{1}{2}$ particle}\label{sec:Pauli}

The aim of this section is to obtain a hydrodynamic, or Madelung-type, formulation of the Pauli equation for a nonrelativistic spin-$\frac12$ particle, and to examine how solutions of the resulting hydrodynamic system give rise to corresponding spinor-field solutions of the standard Pauli equation. 
The hydrodynamic formulation of the Pauli equation is presented in \cite{Tak54}, but without derivation. Since
the details of such derivation are  
crucial to understand the role played by the quantum terms in the corresponding hydrodynamic system, we include it in full detail  here. 
Since the derivation is somewhat lengthy, we first state the final result in terms of the physically natural variables $(\rho,\bol{u},\bol{s})$, namely the density, velocity, and spin fields. The remainder of the section is devoted to deriving this system from the component equations obtained by substituting the polar form of the spinor into the Pauli equation.

\subsection{Hydrodynamic form of the Pauli equation and propagration of the Madelung correspondence}
The Pauli equation for the spinor \(\Psi\) is
\eq{
i\hbar\,\frac{\p\Psi}{\p t}
=
H\Psi,\label{Pauli}
}
where \(H\) denotes the differential (Hamiltonian) operator
\eq{
H=
\frac{1}{2m}
\left(-i\hbar\nabla-q\bol{A}\right)^2
+q\Phi
-\frac{q\hbar}{2m}\,\boldsymbol{\sigma}\cdot \bol{B}.\label{PauliHop}
}
As usual \(m\) and \(q\) denote the particle mass and charge, respectively, \(\lr{\bol{A}\lr{\bol{x},t},\Phi\lr{\bol{x},t}}\) are the electromagnetic potentials, \(\bol{B}=\nabla\cp\bol{A}\) is the external magnetic field, 
and \(\boldsymbol{\sigma}=\lr{\sigma_1,\sigma_2,\sigma_3}\) are the Pauli matrices. 
The following proposition summarizes the hydrodynamic form of the Pauli equation in the variables
$(\rho,\bol{u},\bol{s})$. 

\begin{proposition}[Hydrodynamic form of the Pauli equation]\label{prop:Paulihydro}
Let
$\Psi=\Psi(\bol{x},t):\Omega\times[0,\infty)\to \mathbb{C}^2$ 
be a smooth solution of the Pauli equation \eqref{Pauli}. 
Assume $\rho>0$. Then the fields $(\rho,\bol{u},\bol{s})$,  defined in equations \eqref{rho}, \eqref{u}, and \eqref{s}, 
satisfy the continuity equation
\begin{equation}
\frac{\partial \rho}{\partial t}+\nabla\cdot(\rho\bol{u})=0,
\label{eq:hydro-continuity}
\end{equation}
the momentum equation
\begin{equation}
m\rho\lr{\frac{\partial \bol{u}}{\partial t}+\bol{u}\cdot\nabla\bol{u}}
=
q\rho(\bol{E}+\bol{u}\times\bol{B})
+\frac{\hbar^2}{2m}\rho\,\nabla\!\lr{\frac{\Delta\sqrt{\rho}}{\sqrt{\rho}}} 
+\frac{q\hbar}{2m}\rho 
s_k\nabla B^k
-
\nabla\cdot\Pi,
\label{eq:hydro-momentum}
\end{equation}
where $\Pi$ is a symmetric contravariant spin-stress $2$-tensor with components 
\begin{equation}
\Pi^{ij}
=
\frac{\hbar^2\rho}{4m}\,\partial_i\bol{s}\cdot\partial_j\bol{s},
\label{eq:spin-stress-prop}
\end{equation}
and the spin evolution equation
\begin{equation}
\frac{\partial \bol{s}}{\partial t}
+(\bol{u}\cdot\nabla)\bol{s}
=
\bol{s}\cp\lrc{\frac{q\bol{B}}{m}
+\frac{\hbar}{2m}\lrs{\lr{\nabla\log \rho\cdot\nabla}\bol{s}+\Delta\bol{s}}}.
\label{eq:hydro-spin}
\end{equation}
\end{proposition}

Equation \eqref{eq:hydro-continuity} is the continuity equation for the density, equation \eqref{eq:hydro-momentum} is the Euler-type momentum equation with Lorentz force, quantum (Bohm) force,  magnetic-moment force, and spin stress, while equation \eqref{eq:hydro-spin} describes the advection and precession of the spin field. 
The proof of proposition \ref{prop:Paulihydro} is given in subsection \ref{sec:proofofprop3.1}. 

We may now apply Proposition~\ref{prop:Paulihydro} to show that the hydrodynamic system
\eqref{eq:hydro-continuity}, \eqref{eq:hydro-momentum}, and \eqref{eq:hydro-spin}
propagates the Madelung correspondence: 
\begin{proposition}[Global propagation of the Madelung correspondence]\label{prop:GlobalSpinor}
Let \(\Psi\) be a smooth solution of the Pauli equation \eqref{Pauli} on \([0,T]\), \(T>0\), and assume that its density
$\rho=\Psi^\dagger \Psi$ 
is strictly positive on \(\Omega\times[0,T]\), so that \(\mc M(\Psi)\) is well defined.
Let
\eq{
\Psi_0=\Psi(\bol{x},0),
\qquad
\mc M(\Psi_0)=(\rho_0,\bol{u}_0,\bol{s}_0).
}
Assume moreover that the hydrodynamic Pauli system
\eqref{eq:hydro-continuity}, \eqref{eq:hydro-momentum}, and \eqref{eq:hydro-spin}
admits a unique smooth solution \((\rho,\bol{u},\bol{s})\) on \([0,T]\) with initial data
\eq{
(\rho,\bol{u},\bol{s})|_{t=0}=(\rho_0,\bol{u}_0,\bol{s}_0).
}
Then
\eq{
\mc M(\Psi)=(\rho,\bol{u},\bol{s})
\qquad \rm{on }~~\Omega\times[0,T].
}
In particular, once the initial spinor \(\Psi_0\) is fixed, the corresponding hydrodynamic evolution determines the Pauli spinor evolution uniquely up to an overall constant phase factor.
\end{proposition}

\begin{proof}
By Proposition~\ref{prop:Paulihydro}, since \(\Psi\) is a smooth solution of the Pauli equation and \(\rho>0\), its Madelung transform
\eq{
(\rho',\bol{u}',\bol{s}')=\mc M(\Psi)
}
is a smooth solution of the hydrodynamic system
\eqref{eq:hydro-continuity}, \eqref{eq:hydro-momentum}, and \eqref{eq:hydro-spin}
on \([0,T]\). Moreover, at \(t=0\) one has
\eq{
(\rho',\bol{u}',\bol{s}')|_{t=0}
=
\mc M(\Psi_0)
=
(\rho_0,\bol{u}_0,\bol{s}_0).
}
By the assumed uniqueness of hydrodynamic solutions with these initial data, it follows that
\eq{
(\rho',\bol{u}',\bol{s}')=(\rho,\bol{u},\bol{s})
\qquad \rm{on }~~\Omega\times[0,T].
}
Hence
\eq{
\mc M(\Psi)=(\rho,\bol{u},\bol{s}),
}
as claimed.
\end{proof}

\begin{remark}[Reconstruction of the spinor field from \((\rho,\bol{u},\bol{s})\)]\label{rem:GlobalSpinor}
Under the assumptions of Proposition~\ref{prop:GlobalSpinor}, the spinor field \(\Psi\) can be reconstructed globally from the hydrodynamic fields \((\rho,\bol{u},\bol{s})\), uniquely up to an overall constant phase factor. 
Indeed, on regions where \(s^1\neq 0\), one may recover the angles \(\eta\) and \(\varphi\) from
\eq{
\eta=\arccos\lr{s^3},
\qquad
\varphi=\arctan\lr{\frac{s^2}{s^1}}.
}
If 
\(s^2\neq 0\), one similarly has
\eq{
\eta=\arccos\lr{s^3},
\qquad
\varphi={\rm arccot}\lr{\frac{s^1}{s^2}}.
}
In both cases, one then reconstructs, up to constant phase factors, the phases \(\theta_i\), \(i=1,2\), by integrating
\eq{
\nabla\theta_1
=
m\bol{u}+q\bol{A}
-
\sin^2\lr{\frac{\eta}{2}}\hbar\nabla\varphi,
\qquad
\nabla\theta_2
=
m\bol{u}+q\bol{A}
+
\cos^2\lr{\frac{\eta}{2}}\hbar\nabla\varphi.
}
The spinor field is then recovered as
\begin{equation}
\Psi=
\begin{pmatrix}
\sqrt{\rho}\cos\lr{\frac{\eta}{2}}\exp\lrc{{\rm i}\frac{\theta_1}{\hbar}}\\[1mm]
\sqrt{\rho}\sin\lr{\frac{\eta}{2}}\exp\lrc{{\rm i}\frac{\theta_2}{\hbar}}
\end{pmatrix}.
\end{equation} 
On the other hand, if \(s^1=s^2=0\), then \(s^3=\pm1\), so that \(\eta=\lrc{0,\pi}\), and either \(\rho_2\) or \(\rho_1\) vanishes identically. For example, if \(s^3=1\), then \(\eta=0\), and
\eq{
\nabla\theta_1=m\bol{u}+q\bol{A},
}
while the corresponding spinor takes the form
\eq{
\Psi_1=\sqrt{\rho}\exp\lrc{{\rm i}\frac{\theta_1}{\hbar}},
\qquad
\Psi_2=0,
}
corresponding to a spin-up configuration. Similarly, if \(s^3=-1\), then \(\eta=\pi\), and one obtains a spin-down configuration
\eq{
\Psi_1=0,
\qquad
\Psi_2=\sqrt{\rho}\exp\lrc{{\rm i}\frac{\theta_2}{\hbar}},
\qquad
\nabla\theta_2=m\bol{u}+q\bol{A}.
} 
Thus, although the formulas for \(\eta\) and \(\varphi\) are written in local coordinates, Proposition~\ref{prop:GlobalSpinor} implies that the locally reconstructed expressions are all compatible with the unique Pauli solution determined by the initial spinor \(\Psi_0\). Hence they patch together to yield a globally defined spinor field \(\Psi\), unique up to an overall constant phase factor.
\end{remark}

\begin{remark}[Interpretation of Proposition~\ref{prop:GlobalSpinor}]
Proposition~\ref{prop:GlobalSpinor} should be interpreted as a propagation-of-compatibility result.
Section~2 shows that, in general, the inverse reconstruction of a spinor from \((\rho,\bol{u})\) is only local
and need not be unique. The additional spin field \(\bol{s}\), together with the initial spinor \(\Psi_0\) and the
uniqueness of the hydrodynamic evolution, removes this ambiguity along the time evolution.
Accordingly, the proposition does not assert that an arbitrary hydrodynamic triple \((\rho,\bol{u},\bol{s})\)
admits a globally defined inverse \(\mc M^{-1}\) in complete generality; rather, it states that if the initial
hydrodynamic data come from a spinor field \(\Psi_0\), then the unique hydrodynamic solution remains
exactly the Madelung transform of the corresponding Pauli solution for all \(t\in[0,T]\). A schematic illustration of this correspondence is given in Figure~\ref{fig:MadelungPropagation}.
\end{remark}


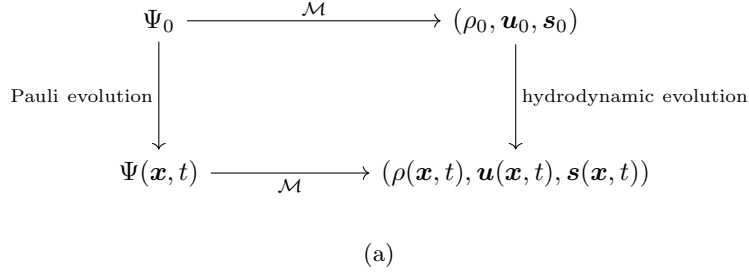
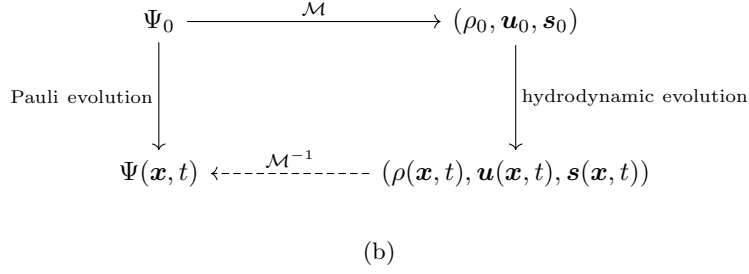
\begin{figure}[t]
\centering

\begin{subfigure}[t]{0.95\textwidth}
\centering
\[
\begin{tikzcd}[column sep=6em, row sep=4em]
\Psi_0
\arrow[r, "\mc M"]
\arrow[d, "{\rm Pauli\ evolution}"']
&
(\rho_0,\bol{u}_0,\bol{s}_0)
\arrow[d, "{\rm hydrodynamic\ evolution}"]
\\
\Psi(\bol{x},t)
\arrow[r, "\mc M"']
&
(\rho(\bol{x},t),\bol{u}(\bol{x},t),\bol{s}(\bol{x},t))
\end{tikzcd}
\]
\caption{}
\label{fig:MadelungPropagation-top}
\end{subfigure}

\vspace{1em}

\begin{subfigure}[t]{0.95\textwidth}
\centering
\[
\begin{tikzcd}[column sep=6em, row sep=4em]
\Psi_0
\arrow[r, "\mc M"]
\arrow[d, "{\rm Pauli\ evolution}"']
&
(\rho_0,\bol{u}_0,\bol{s}_0)
\arrow[d, "{\rm hydrodynamic\ evolution}"]
\\
\Psi(\bol{x},t)
&
(\rho(\bol{x},t),\bol{u}(\bol{x},t),\bol{s}(\bol{x},t))
\arrow[l, dashed, "\mc M^{-1}"']
\end{tikzcd}
\]
\caption{}
\label{fig:MadelungPropagation-bottom}
\end{subfigure}

\caption{
Schematic illustration of the Madelung correspondence.
Top: starting from compatible initial data, evolution by the Pauli equation on the spinor side and evolution by the hydrodynamic system on the fluid side commute with the Madelung transform \(\mc M\).
Bottom: the top arrow emphasizes that the initial hydrodynamic data are induced from the initial spinor \(\Psi_0\) by \(\mc M\), while the dashed lower arrow indicates the inverse reconstruction in the sense of Remark~\ref{rem:GlobalSpinor}; once \(\Psi_0\) is fixed, the compatible branch of \(\mc M^{-1}\) reconstructs the Pauli spinor from the hydrodynamic variables, uniquely up to an overall constant phase factor.
}
\label{fig:MadelungPropagation}
\end{figure}


\subsection{Proof of proposition~\ref{prop:Paulihydro}}\label{sec:proofofprop3.1}


\medskip
\noindent 
The proof of proposition~\ref{prop:Paulihydro} involves three main steps: 
starting from the real and imaginary parts of the spinor form of the Pauli equation \eqref{Pauli}, we (i) obtain the continuity equation, (ii) derive the momentum equation in the material-derivative form \eqref{eq:hydro-momentum}, and (iii) obtain the spin equation \eqref{eq:hydro-spin} by taking suitable linear combinations of the phase and continuity equations.

\medskip
\noindent (i) \underline{\ti{Derivation of the continuity equation.}}
We begin by expressing the Pauli equation~\eqref{Pauli} in terms of its components $\Psi_1$ and $\Psi_2$. Explicitly, equation~\eqref{Pauli} reads as
\sys{
{\rm i}\hbar\frac{\p\Psi_1}{\p t}=&\lrs{\frac{1}{2m}
\lr{-\hbar^2\Delta+{\rm i}\hbar q\nabla\cdot\bol{A}+2{\rm i}\hbar q\bol{A}\cdot\nabla+q^2\bol{A}^2}+q\Phi
}\Psi_1-\frac{q\hbar}{2m}
\lrs{B^3\Psi_1+\lr{B^1-{\rm i}B^2}\Psi_2},\\
{\rm i}\hbar\frac{\p\Psi_2}{\p t}=&\lrs{\frac{1}{2m}
\lr{-\hbar^2\Delta+{\rm i}\hbar q\nabla\cdot\bol{A}+2{\rm i}\hbar q\bol{A}\cdot\nabla+q^2\bol{A}^2}+q\Phi
}\Psi_2
+\frac{q\hbar}{2m}\lrs{B^3\Psi_2-\lr{B^1+{\rm i}B^2}\Psi_1}.
}{Pauli_}
Substituting the Madelung form of the spinor field \eqref{spinor} 
into the Pauli equation \eqref{Pauli_}, and defining the quantum potential energy of particle $i$ as  
\eq{
V_{qi}=-\frac{\hbar^2}{2m}\frac{\Delta\sqrt{\rho_i}}{\sqrt{\rho_i}}=-\frac{\hbar^2}{2m}\lr{\frac{\Delta\rho_i}{2\rho_i}-\frac{\abs{\nabla\rho_i}^2}{4\rho_i^2}},\qquad i=1,2,} one obtains the equations
\eq{\label{Pauli1}
\lr{\frac{{\rm i}\hbar}{2\rho_1}\frac{\p\rho_1}{\p t}-\frac{\p\theta_1}{\p t}}\Psi_1=&\lr{V_{q1}+\frac{\abs{\nabla\theta_1}^2}{2m}+q\Phi-\frac{q\bol{A}\cdot\nabla\theta_1}{m}+\frac{q^2\bol{A}^2}{2m}}\Psi_1\\
&+\frac{{\rm i}\hbar}{2m}\lr{
-\frac{\nabla\rho_1\cdot\nabla\theta_1}{\rho_1}-\Delta\theta_1
+\frac{q\bol{A}\cdot\nabla\rho_1}{\rho_1}+q\nabla\cdot\bol{A}}\Psi_1\\
&-\frac{q\hbar}{2m}\lrs{B^3\Psi_1+\lr{B^1-{\rm i}B^2}\Psi_2},
}
and 
\eq{
\lr{\frac{{\rm i}\hbar}{2\rho_2}\frac{\p\rho_2}{\p t}-\frac{\p\theta_2}{\p t}}\Psi_2=&\lr{V_{q2}+\frac{\abs{\nabla\theta_2}^2}{2m}+q\Phi-\frac{q\bol{A}\cdot\nabla\theta_2}{m}+\frac{q^2\bol{A}^2}{2m}}\Psi_2\\
&+\frac{{\rm i}\hbar}{2m}\lr{
-\frac{\nabla\rho_2\cdot\nabla\theta_2}{\rho_2}-\Delta\theta_2
+\frac{q\bol{A}\cdot\nabla\rho_2}{\rho_2}+q\nabla\cdot\bol{A}}\Psi_1\\
&+\frac{q\hbar}{2m}\lrs{B^3\Psi_2-\lr{B^1+{\rm i}B^2}\Psi_1}.\label{Pauli2}
} 
Separating real and imaginary parts, equations \eqref{Pauli1} and \eqref{Pauli2} are equivalent to the system of equations
\eq{
\frac{\p\rho_1}{\p t}=-\frac{1}{m}\nabla\cdot\lrs{\rho_1\lr{\nabla\theta_1-q\bol{A}}}
-\frac{q\sqrt{\rho_1\rho_2}}{m}\lr{
B^1\sin\varphi-B^2\cos\varphi
}
,\label{rho1}
}
\eq{
\frac{\p\theta_1}{\p t}=&-
{V_{q1}
-\frac{1}{2m}\lr{\nabla\theta_1-q\bol{A}}^2
-q\Phi
}+\frac{q\hbar}{2m}\lrs{
B^3+\sqrt{\frac{\rho_2}{\rho_1}}\lr{B^1\cos\varphi
+
B^2\sin\varphi
}},\label{theta1} 
}
\eq{
\frac{\p\rho_2}{\p t}=-\frac{1}{m}\nabla\cdot\lrs{\rho_2\lr{\nabla\theta_2-q\bol{A}}}
+\frac{q\sqrt{\rho_1\rho_2}}{m}\lr{B^1\sin\varphi-B^2\cos\varphi
},\label{rho2}
}
\eq{
\frac{\p\theta_2}{\p t}=&-
{V_{q2}-\frac{1}{2m}\lr{\nabla\theta_2-q\bol{A}}^2
-q\Phi
}+\frac{q\hbar}{2m}\lrs{-
B^3+\sqrt{\frac{\rho_1}{\rho_2}}\lr{B^1\cos\varphi
+
B^2\sin\varphi}
}.\label{theta2}
}
Combining equations \eqref{rho1} and \eqref{rho2}, 
and recalling equations  \eqref{rho} and \eqref{u}, 
one obtains the continuity equation for the density $\rho$, 
\eq{
\frac{\p\rho}{\p t}=-\frac{1}{m}\nabla\cdot\lr{\rho_1\nabla\theta_1+\rho_2\nabla\theta_2-q\rho\bol{A}}=-\nabla\cdot\lr{\rho\bol{u}}.\label{continuity}
}

\medskip
\noindent (ii) \underline{\ti{Derivation of the momentum equation.}} 
To simplify the analysis leading to the momentum equation, it is convenient to introduce several quantities. 
In terms of $\eta$, the spinor field $\Psi$ takes the form
\begin{equation}\label{spinor2}
\Psi=\begin{pmatrix}\Psi_1\\\Psi_2\end{pmatrix}
=\sqrt{\rho}\begin{pmatrix} \cos\lr{\frac{\eta}{2}}\exp\lrc{{\rm i}\frac{\theta_1}{\hbar}}\\\sin\lr{\frac{\eta}{2}}\exp\lrc{{\rm i}\frac{\theta_2}{\hbar}}\end{pmatrix}.
\end{equation} 
We introduce the kinetic  momenta 
\eq{
\bol{p}_i=\nabla\theta_i-q\bol{A},\qquad i=1,2.\label{pi}
}
We have
\eq{
m\rho\bol{u}=\rho_1\bol{p}_1+\rho_2\bol{p}_2,\qquad\bol{p}_2-\bol{p}_1=\hbar\nabla\varphi,\qquad \nabla\cp\bol{p}_i=-q\bol{B},\qquad i=1,2.\label{p1p2prop}
}
It is also convenient to introduce the magnetic potential energies,
\sys{
M_1=&-\frac{q\hbar}{2m}\lrs{
B^3+\sqrt{\frac{\rho_2}{\rho_1}}\lr{B^1\cos\varphi
+
B^2\sin\varphi
}},\label{M1}\\
M_2=&-\frac{q\hbar}{2m}\lrs{-
B^3+\sqrt{\frac{\rho_1}{\rho_2}}\lr{B^1\cos\varphi
+
B^2\sin\varphi}
},\label{M2}
}{M12}
and notice that the last terms of equations \eqref{rho1} and \eqref{rho2} can be expressed as 
\eq{
\frac{q\sqrt{\rho_1\rho_2}}{m}\lr{B^1\sin\varphi-B^2\cos\varphi}=\frac{q\rho}{2m}\lr{B^1s^2-B^2s^1}=\frac{q\rho}{2m}\bol{B}\cp\bol{s}\cdot\nabla z=\frac{\rho\omega_z}{2},
}
where $\omega_z=q\bol{B}\cp\bol{s}\cdot\nabla z/m$ has dimensions of frequency. 
We thus have
\eq{\label{mom}
m\rho\frac{\p\bol{u}}{\p t}=&-
\frac{\p\rho}{\p t}\lr{m\bol{u}+q\bol{A}}+
{\frac{\p\rho_1}{\p t}\nabla\theta_1+\rho_1\nabla\frac{\p\theta_1}{\p t}+\frac{\p\rho_2}{\p t}\nabla\theta_2+\rho_2\nabla\frac{\p\theta_2}{\p t}}-q\rho\frac{\p\bol{A}}{\p t}\\
=&
\frac{\hbar}{\rho}\lr{\rho_1\frac{\p\rho_2}{\p t}-\rho_2\frac{\p\rho_1}{\p t}}\nabla\varphi 
-\rho_1\nabla\lr{V_{q1}+\frac{\bol{p}_1^2}{2m}+M_1}
-\rho_2\nabla\lr{V_{q2}+\frac{\bol{p}_2^2}{2m}+M_2}-q\rho\nabla\Phi
-q\rho\frac{\p\bol{A}}{\p t}\\
=&
\frac{\hbar}{m\rho}\lrs{\rho_2\nabla\cdot\lr{\rho_1\bol{p}_1}-
\rho_1\nabla\cdot\lr{\rho_2\bol{p}_2}
}\nabla\varphi
+\frac{\hbar \rho\omega_z}{2}\nabla\varphi-\frac{\rho_1}{2m}\nabla\bol{p}_1^2
-\frac{\rho_2}{2m}\nabla\bol{p}_2^2\\
&-\rho_1\nabla\lr{V_{q1}+M_1}-\rho_2\nabla\lr{V_{q2}+M_2}+q\rho\bol{E}.
}
Next, we rewrite the $\nabla \bol{p}_i^{\,2}$ terms in divergence form. Let
\begin{equation}
D_i=\nabla\cdot(\rho_i\bol{p}_i),\qquad i=1,2.
\end{equation}
Using
\begin{equation}
\nabla\cdot(\rho_i \bol{p}_i\otimes \bol{p}_i)
=
D_i \bol{p}_i+\rho_i(\bol{p}_i\cdot\nabla)\bol{p}_i
\end{equation}
together with
\begin{equation}
(\bol{p}_i\cdot\nabla)\bol{p}_i
=
\nabla\!\lr{\frac{\bol{p}_i^{\,2}}{2}}
-
\bol{p}_i\times(\nabla\times\bol{p}_i),
\end{equation}
we obtain
\begin{equation}
-\frac{\rho_i}{2m}\nabla \bol{p}_i^{\,2}
=
-\frac{1}{m}\nabla\cdot(\rho_i\bol{p}_i\otimes\bol{p}_i)
+\frac{1}{m}D_i\bol{p}_i
-\frac{\rho_i}{m}\bol{p}_i\times(\nabla\times\bol{p}_i).
\end{equation}
Since $\nabla\times\bol{p}_i=-q\bol{B}$, this becomes
\begin{equation}
-\frac{\rho_i}{2m}\nabla \bol{p}_i^{\,2}
=
-\frac{1}{m}\nabla\cdot(\rho_i\bol{p}_i\otimes\bol{p}_i)
+\frac{1}{m}D_i\bol{p}_i
+\frac{q}{m}\rho_i \bol{p}_i\times\bol{B}.
\end{equation}
Summing over $i=1,2$ and substituting into \eqref{mom}, we find
\eq{
m\rho \frac{\partial \bol{u}}{\partial t}
=&
-\frac{1}{m}\nabla\cdot\lr{\rho_1\bol{p}_1\otimes\bol{p}_1+\rho_2\bol{p}_2\otimes\bol{p}_2}
+\frac{1}{m}(D_1\bol{p}_1+D_2\bol{p}_2)
+\frac{\hbar}{m\rho}(\rho_2D_1-\rho_1D_2)\nabla\varphi\\
&+\frac{q}{m}(\rho_1\bol{p}_1+\rho_2\bol{p}_2)\times\bol{B}
+\frac{\hbar\rho\omega_z}{2}\nabla\varphi
-\rho_1\nabla(V_{q1}+M_1)-\rho_2\nabla(V_{q2}+M_2)+q\rho \bol{E}\\
=&
-\frac{1}{m}\nabla\cdot\lr{\rho_1\bol{p}_1\otimes\bol{p}_1+\rho_2\bol{p}_2\otimes\bol{p}_2}
+\frac{1}{m}(D_1\bol{p}_1+D_2\bol{p}_2)
+\frac{\hbar}{m\rho}(\rho_2D_1-\rho_1D_2)\nabla\varphi\\
&
+\frac{\hbar\rho\omega_z}{2}\nabla\varphi
-\rho_1\nabla(V_{q1}+M_1)-\rho_2\nabla(V_{q2}+M_2)+q\rho\lr{\bol{E}+\bol{u}\cp\bol{B}}.
\label{eq:momentum-1}
}
We now simplify the terms containing $D_1$ and $D_2$. Using~\eqref{p1p2prop} 
we have
\begin{align}
&\frac{1}{m}(D_1\bol{p}_1+D_2\bol{p}_2)
+\frac{\hbar}{m\rho}(\rho_2D_1-\rho_1D_2)\nabla\varphi
\nonumber\\
&\qquad=
\frac{1}{m}
\left[
D_1\bol{p}_1+D_2\bol{p}_2
+\frac{\rho_2D_1-\rho_1D_2}{\rho}(\bol{p}_2-\bol{p}_1)
\right]
\nonumber\\
&\qquad=
\frac{D_1+D_2}{m\rho}(\rho_1\bol{p}_1+\rho_2\bol{p}_2)
=
(D_1+D_2)\bol{u}.
\end{align}
By the continuity equation \eqref{continuity},
\begin{equation}
D_1+D_2=-m\frac{\partial \rho}{\partial t},
\end{equation}
and therefore
\begin{equation}
\frac{1}{m}(D_1\bol{p}_1+D_2\bol{p}_2)
+\frac{\hbar}{m\rho}(\rho_2D_1-\rho_1D_2)\nabla\varphi
=
-m\frac{\partial \rho}{\partial t}\,\bol{u}.
\end{equation}
Hence \eqref{eq:momentum-1} becomes
\begin{equation}
\frac{\partial}{\partial t}(m\rho \bol{u})
=
-\frac{1}{m}\nabla\cdot\lr{\rho_1\bol{p}_1\otimes\bol{p}_1+\rho_2\bol{p}_2\otimes\bol{p}_2}
+q\rho(\bol{E}+\bol{u}\times\bol{B})
+\frac{\hbar\rho\omega_z}{2}\nabla\varphi
-\rho_1\nabla(V_{q1}+M_1)-\rho_2\nabla(V_{q2}+M_2).
\label{eq:momentum-conservative}
\end{equation}
We next rewrite the flux term. Writing
\begin{equation}
\alpha=\frac{\rho_1}{\rho},
\qquad
\beta=\frac{\rho_2}{\rho},
\qquad
\alpha+\beta=1,
\end{equation}
we have
\begin{equation}
\bol{p}_1=m\bol{u}- \hbar \beta\nabla\varphi,
\qquad
\bol{p}_2=m\bol{u}+ \hbar \alpha\nabla\varphi.\label{p1p2uphi}
\end{equation}
A direct computation gives
\begin{equation}
\rho_1\bol{p}_1\otimes\bol{p}_1+\rho_2\bol{p}_2\otimes\bol{p}_2
=
m^2\rho\,\bol{u}\otimes\bol{u}
+
\rho\,\alpha\beta\,\hbar^2\nabla\varphi\otimes\nabla\varphi.
\end{equation}
Since
\begin{equation}
\alpha\beta=\frac{\rho_1\rho_2}{\rho^2}=\frac14\sin^2\eta,
\end{equation}
it follows that
\begin{equation}
\rho_1\bol{p}_1\otimes\bol{p}_1+\rho_2\bol{p}_2\otimes\bol{p}_2
=
m^2\rho\,\bol{u}\otimes\bol{u}
+
\frac{\rho\hbar^2}{4}\sin^2\eta\,\nabla\varphi\otimes\nabla\varphi.
\label{eq:flux-split}
\end{equation}
We now simplify the potential terms. First, from the definitions \eqref{M1} and \eqref{M2} of $M_1$ and $M_2$,
\begin{equation}
\rho_1M_1+\rho_2M_2
=
-\frac{q\hbar}{2m}\,\rho\,\bol{B}\cdot\bol{s}.
\label{eq:M-combined}
\end{equation}
Second, we notice the following identity, 
\begin{equation}
\rho_1V_{q1}+\rho_2V_{q2}
=
-\frac{\hbar^2}{2m}\rho\frac{\Delta\sqrt{\rho}}{\sqrt{\rho}}
+
\frac{\hbar^2}{8m}\rho \abs{\nabla\eta}^2.
\label{eq:Vq-combined}
\end{equation}
We introduce the quantum (Bohm) potential
\begin{equation}
Q=-\frac{\hbar^2}{2m}\frac{\Delta\sqrt{\rho}}{\sqrt{\rho}},\label{Q}
\end{equation}
and express equation \eqref{eq:momentum-conservative} as follows:
\eq{
\frac{\p}{\p t}\lr{m\rho\bol{u}}=&-m\nabla\cdot\lr{\rho\bol{u}\otimes\bol{u}}+q\rho\lr{\bol{E}+\bol{u}\cp\bol{B}}-\frac{\hbar^2}{4m}\nabla\cdot\lr{\rho\sin^2\eta\nabla\varphi\otimes\nabla\varphi}\\
&+\frac{\hbar\rho\omega_z}{2}\nabla\varphi+\nabla\lr{\frac{q\hbar}{2m}\rho\bol{B}\cdot\bol{s}-\rho Q-\frac{\hbar^2}{8m}\rho\abs{\nabla\eta}^2}+\lr{V_{q_1}+M_1}\nabla\rho_1+\lr{V_{q_2}+M_2}\nabla\rho_2\\
=&-m\nabla\cdot\lr{\rho\bol{u}\otimes\bol{u}}+q\rho\lr{\bol{E}+\bol{u}\cp\bol{B}}-\rho\nabla Q+\frac{q\hbar}{2m}\rho
s_k\nabla B^k
\\
&+\frac{q\hbar}{2m}\rho B^k\nabla s_k-\frac{\hbar^2}{4m}\nabla\cdot\lr{\rho\sin^2\eta\nabla\varphi\otimes\nabla\varphi}
+\frac{\hbar\rho\omega_z}{2}\nabla\varphi
+\lr{\frac{q\hbar}{2m}\bol{B}\cdot\bol{s}-Q}\nabla\rho\\&-\frac{\hbar^2}{8m}\nabla\lr{\rho\abs{\nabla\eta}^2}
+\lr{V_{q_1}+M_1}\nabla\rho_1+\lr{V_{q_2}+M_2}\nabla\rho_2.\label{eq:momentum_1}
}
The remaining task is to express the last two lines of this equation as the divergence of  the spin-stress tensor
\begin{equation}
\Pi^{ij}
=
\frac{\hbar^2\rho}{4m}\,\partial_i\bol{s}\cdot\partial_j\bol{s}
=
\frac{\hbar^2\rho}{4m}
\lr{
\partial_i\eta\,\partial_j\eta
+
\sin^2\eta\,\partial_i\varphi\,\partial_j\varphi
}.
\label{eq:Pi-spin}
\end{equation}
We have
\eq{
\nabla\cdot\Pi=\frac{\hbar^2}{4m}\lrs{\nabla\cdot\lr{\rho\nabla\eta}\nabla\eta
+\frac{\rho}{2}\nabla\abs{\nabla\eta}^2
+\nabla\cdot\lr{\rho\sin^2\eta\nabla\varphi}\nabla\varphi
+\frac{\rho\sin^2\eta}{2}\nabla\abs{\nabla\varphi}^2
}.\label{divPi}
}
Substituting~\eqref{divPi} 
into~\eqref{eq:momentum_1}, 
denoting $d/dt=\p/\p t+\bol{u}\cdot\nabla$, 
and using the continuity equation~\eqref{continuity}, we obtain
\eq{
m\rho\frac{d\bol{u}}{dt}=&
q\rho\lr{\bol{E}+\bol{u}\cp\bol{B}}-\rho\nabla Q+\frac{q\hbar}{2m}\rho s_k\nabla B^k-\nabla\cdot\Pi\\
&+\frac{q\hbar}{2m}\rho B^k\nabla s_k
+\frac{\hbar\rho\omega_z}{2}\nabla\varphi
+\lr{\frac{q\hbar}{2m}\bol{B}\cdot\bol{s}-Q-\frac{\hbar^2}{8m}\abs{\nabla\eta}^2}\nabla\rho\\&+\lr{V_{q_1}+M_1}\nabla\rho_1+\lr{V_{q_2}+M_2}\nabla\rho_2
+\frac{\hbar^2}{4m}\nabla\cdot\lr{\rho\nabla\eta}\nabla\eta.\label{eq:momentum_2}
}
Let us show that the last two lines vanish. We proceed by eliminating the components of this expression along $\nabla\rho$, $\nabla\varphi$, and $\nabla\eta$ one by one. 
Using~\eqref{eta}, 
we obtain
\eq{
\nabla \rho_1
=
\alpha \nabla \rho
-\frac{\rho}{2}\sin\eta\,\nabla\eta,
\qquad
\nabla \rho_2
=
\beta \nabla \rho
+\frac{\rho}{2}\sin\eta\,\nabla\eta .
}
Therefore
\begin{align}
&(V_{q1}+M_1)\nabla \rho_1+(V_{q2}+M_2)\nabla \rho_2 \notag\\
&\qquad=
\bigl[\alpha (V_{q1}+M_1)+\beta (V_{q2}+M_2)\bigr]\nabla\rho
+\frac{\rho}{2}\sin\eta
\lr{V_{q2}+M_2-V_{q1}-M_1}\nabla\eta .
\label{eq:split-rho1-rho2}
\end{align}
Using \eqref{eq:M-combined} and \eqref{eq:Vq-combined}, it follows that
\begin{align}
&(V_{q1}+M_1)\nabla \rho_1+(V_{q2}+M_2)\nabla \rho_2 \notag\\
&\qquad=
\left(
Q+\frac{\hbar^2}{8m}|\nabla\eta|^2
-\frac{q\hbar}{2m} \bol{B}\cdot \bol{s}
\right)\nabla\rho
+\frac{\rho}{2}\sin\eta
\lr{V_{q2}-V_{q1}+M_2-M_1}\nabla\eta .
\label{eq:split-rho1-rho2-2}
\end{align}
Substituting this into \eqref{eq:momentum_2}, the terms proportional to $\nabla\rho$ cancel, 
\eq{
m\rho\frac{d\bol{u}}{dt}=&
q\rho\lr{\bol{E}+\bol{u}\cp\bol{B}}-\rho\nabla Q+\frac{q\hbar}{2m}\rho\, s_k\nabla B^k-\nabla\cdot\Pi\\
&+\frac{q\hbar}{2m}\rho B^k\nabla s_k
+\frac{\hbar\rho\omega_z}{2}\nabla\varphi
+\frac{\rho}{2}\sin\eta\lr{V_{q2}-V_{q1}+M_2-M_1}\nabla\eta
+\frac{\hbar^2}{4m}\nabla\cdot\lr{\rho\nabla\eta}\nabla\eta.
\label{eq:momentum_3}
}

\noindent For the $\varphi$-part, we first note that
\eq{
\frac{q\hbar}{2m}\rho B^k\nabla s_k=
\frac{q\hbar}{2m}\rho\lrc{
\sin\eta\lr{B^2\cos\varphi-B^1\sin\varphi}\nabla\varphi
+\lrs{
\cos\eta\lr{B^1\cos\varphi+B^2\sin\varphi}
-B^3\sin\eta
}\nabla\eta
},
}
and that
\eq{
\frac{\hbar\rho\omega_z}{2}\nabla\varphi=\frac{q\hbar}{2m}\rho\sin\eta\lr{B^1\sin\varphi-B^2\cos\varphi}\nabla\varphi.
}
The terms along $\nabla\varphi$ in the last line of \eqref{eq:momentum_3} therefore vanish, and we arrive at
\eq{
m\rho\frac{d\bol{u}}{dt}=&
q\rho\lr{\bol{E}+\bol{u}\cp\bol{B}}-\rho\nabla Q+\frac{q\hbar}{2m}\rho\, s_k\nabla B^k-\nabla\cdot\Pi\\
&+\frac{q\hbar}{2m}\rho\lrs{
\cos\eta\lr{B^1\cos\varphi+B^2\sin\varphi}
-B^3\sin\eta
}
\nabla\eta
+\frac{\rho}{2}\sin\eta\lr{V_{q2}-V_{q1}+M_2-M_1}\nabla\eta\\&+\frac{\hbar^2}{4m}\nabla\cdot\lr{\rho\nabla\eta}\nabla\eta.
\label{eq:momentum_4}
}

\noindent For the $\eta$-part, we begin by noting that
\eq{
\frac{\rho}{2}\sin\eta\lr{M_2-M_1}=\frac{q\hbar}{2m}\rho\lrs{B^3\sin\eta-\cos\eta\lr{B^1\cos\varphi+B^2\sin\varphi}}.
}
The only remaining terms to be handled in the last two lines of eq.   \eqref{eq:momentum_4} 
are therefore
\eq{
\frac{\rho}{2}\sin\eta\lr{V_{q2}-V_{q1}}+\frac{\hbar^2}{4m}\nabla\cdot\lr{\rho\nabla\eta}=0,
}
which vanishes due to  the key identity for the difference of the two quantum potentials. 
Noting that $\nabla\alpha=-\nabla\beta=-\sin\eta\nabla\eta/2$, we have
\eq{
V_{q2}-V_{q1}
=&\frac{\hbar^2}{4m}\lr{\frac{\nabla\alpha\cdot\nabla\rho}{\alpha\rho}+\frac{\Delta\alpha}{\alpha}-\frac{\abs{\nabla\alpha}^2}{2\alpha^2}
-\frac{\nabla\beta\cdot\nabla\rho}{\beta\rho}-\frac{\Delta\beta}{\beta}+\frac{\abs{\nabla\beta}^2}{2\beta^2}
}\\
=&\frac{\hbar^2}{4m}
\lrs{
\lr{\frac{1}{\alpha}+\frac{1}{\beta}}\frac{\nabla\rho\cdot\nabla\alpha}{\rho}
+\lr{\frac{1}{\alpha}+\frac{1}{\beta}}\Delta\alpha
+\lr{\frac{1}{\beta^2}-\frac{1}{\alpha^2}}\frac{\abs{\nabla\alpha}^2}{2}}\\
=&\frac{\hbar^2}{4m\alpha\beta\rho}\lrs{
\nabla\rho\cdot\nabla\alpha+\rho\Delta\alpha+\frac{\alpha-\beta}{\alpha\beta}\rho\frac{\abs{\nabla\alpha}^2}{2}
}
\\=&
-\frac{\hbar^2}{2m}\,
\frac{\nabla\cdot(\rho\nabla\eta)}{\rho\sin\eta}.
\label{eq:Vq-difference}
} 
Equation \eqref{eq:momentum_4} therefore becomes 
\begin{equation}
m\rho\frac{d\bol{u}}{dt}
=
q\rho(\bol{E}+\bol{u}\times\bol{B})
+\frac{\hbar^2}{2m}\rho\,\nabla\!\lr{\frac{\Delta\sqrt{\rho}}{\sqrt{\rho}}}
+\frac{q\hbar}{2m}\rho s_k\nabla{B}^k
-\nabla\cdot\Pi.
\label{eq:Takabayasi-momentum}
\end{equation}
This is the desired momentum equation.

\medskip 
\noindent (iii) \underline{\ti{Derivation of the spin evolution equation.}} 
We now derive the evolution equation for the spin field $\bol{s}$. 
Since
$\bol{s}
=
(\sin\eta\cos\varphi,\sin\eta\sin\varphi,\cos\eta)$, 
it is enough to derive evolution equations for $\eta$ and $\varphi$.
We begin with the equation for $\eta$.  
We have
\begin{equation}
\frac{\partial \rho_1}{\partial t}
=
\alpha \frac{\partial \rho}{\partial t}
-
\frac{\rho}{2}\sin\eta\,\frac{\partial \eta}{\partial t},
\qquad
\frac{\partial \rho_2}{\partial t}
=
\beta \frac{\partial \rho}{\partial t}
+
\frac{\rho}{2}\sin\eta\,\frac{\partial \eta}{\partial t}.
\end{equation} 
Hence, 
\eq{
\alpha\frac{\p\rho_2}{\p t}-\beta\frac{\p\rho_1}{\p t}=\frac{\rho}{2}\sin\eta\frac{\p\eta}{\p t}.
}
Recalling \eqref{rho1} and \eqref{rho2}, 
we obtain 
\begin{equation}
\frac{\rho}{2}\sin\eta\,\frac{\partial \eta}{\partial t}
=
\frac{1}{m}\lr{\beta \nabla\cdot(\rho_1\bol{p}_1)-\alpha \nabla\cdot(\rho_2\bol{p}_2)}
+\frac{q\sqrt{\rho_1\rho_2}}{m}\lr{B^1\sin\varphi-B^2\cos\varphi}.
\label{eq:eta-start}
\end{equation}
Now, from \eqref{p1p2uphi}, we may write 
\begin{equation}
\rho_1\bol{p}_1
=
m\rho_1\bol{u}
-
\rho\alpha\beta\hbar\nabla\varphi,
\qquad
\rho_2\bol{p}_2
=
m\rho_2\bol{u}
+
\rho\alpha\beta\hbar\nabla\varphi,
\end{equation}
so
\begin{align}
\beta \nabla\cdot(\rho_1\bol{p}_1)-\alpha \nabla\cdot(\rho_2\bol{p}_2)
&=
m\lrs{\beta \nabla\cdot(\rho_1\bol{u})-\alpha \nabla\cdot(\rho_2\bol{u})}
-\hbar \nabla\cdot(\rho\alpha\beta\nabla\varphi).
\end{align}
Since $\alpha+\beta=1$ and $\beta\nabla\alpha-\alpha\nabla\beta=\nabla\alpha$, one finds
\begin{equation}
\beta \nabla\cdot(\rho_1\bol{u})-\alpha \nabla\cdot(\rho_2\bol{u})
=
\bol{u}\cdot(\beta\nabla\rho_1-\alpha\nabla\rho_2)
=
-\frac{\rho}{2}\sin\eta\,\bol{u}\cdot\nabla\eta.
\end{equation}
Substituting into \eqref{eq:eta-start}, we arrive at
\begin{equation}
\frac{\partial \eta}{\partial t}
+\bol{u}\cdot\nabla\eta
=
-\frac{\hbar}{2m\rho\sin\eta}\nabla\cdot\!\lr{\rho\sin^2\eta\,\nabla\varphi}
+\frac{q}{m}\lr{B^1\sin\varphi-B^2\cos\varphi}.
\label{eq:eta-evolution}
\end{equation}

\noindent We next derive the equation for $\varphi$. Subtracting \eqref{theta1} from \eqref{theta2}, 
we get
\begin{equation}
\hbar\frac{\partial \varphi}{\partial t}
=
-(V_{q2}-V_{q1})
-\frac{1}{2m}\lr{\bol{p}_2^{\,2}-\bol{p}_1^{\,2}}
+\frac{q\hbar}{m}
\left[
-B^3+\cot\eta\lr{B^1\cos\varphi+B^2\sin\varphi}
\right].
\label{eq:phi-start}
\end{equation}
Using \eqref{p1p2uphi}
we compute
\begin{equation}
\frac{\bol{p}_2^{\,2}-\bol{p}_1^{\,2}}{2m}
=
\hbar\,\bol{u}\cdot\nabla\varphi
+\frac{\hbar^2}{2m}\cos\eta\,\abs{\nabla\varphi}^2.
\label{eq:p-difference}
\end{equation}
Substituting \eqref{eq:p-difference} and \eqref{eq:Vq-difference} into \eqref{eq:phi-start}, we obtain
\begin{equation}
\frac{\partial \varphi}{\partial t}
+\bol{u}\cdot\nabla\varphi
=
\frac{\hbar}{2m\rho\sin\eta}\nabla\cdot(\rho\nabla\eta)
-\frac{\hbar}{2m}\cos\eta\,\abs{\nabla\varphi}^2
+\frac{q}{m}
\left[
-B^3+\cot\eta\lr{B^1\cos\varphi+B^2\sin\varphi}
\right].
\label{eq:phi-evolution}
\end{equation}
We now combine \eqref{eq:eta-evolution} and \eqref{eq:phi-evolution} into an equation for $\bol{s}$. Introduce the tangent orthonormal frame $\lr{\bol{s},\bol
s_{\eta},\bol{s}_{\varphi}}$ where 
\begin{equation}
\bol{s}_\eta=\frac{\partial \bol{s}}{\partial \eta}
=
(\cos\eta\cos\varphi,\cos\eta\sin\varphi,-\sin\eta),
\qquad
\bol{s}_\varphi=\frac{1}{\sin\eta}\frac{\partial \bol{s}}{\partial \varphi}
=
(-\sin\varphi,\cos\varphi,0),
\end{equation}
so that
\begin{equation}
\bol{s}\times\bol{s}_\eta=\bol{s}_\varphi,
\qquad
\bol{s}\times\bol{s}_\varphi=-\bol{s}_\eta.
\end{equation}
Hence
\eq{
\frac{d\bol{s}}{dt}=\frac{d\eta}{dt}\bol{s}_{\eta}+\sin\eta\frac{d\varphi}{dt}\bol{s}_{\varphi}. \label{st}
}

\noindent On the other hand, decomposing $\bol{B}=\lr{\bol{B}\cdot\bol{s}}\bol{s}+\lr{\bol{B}\cdot\bol{s}_{\eta}}\bol{s}_{\eta}+\lr{\bol{B}\cdot\bol{s}_{\varphi}}\bol{s}_{\varphi}$ in the frame $(\bol{s},\bol{s}_\eta,\bol{s}_\varphi)$, we have
\begin{equation}
\bol{s}\times\bol{B}
=
\lr{B^1\sin\varphi-B^2\cos\varphi}\bol{s}_\eta
+
\lrs{\cos\eta(B^1\cos\varphi+B^2\sin\varphi)-B^3\sin\eta}\bol{s}_\varphi.
\label{eq:s-cross-B}
\end{equation}
Also, from
\begin{equation}
\partial_i \bol{s}
=
(\partial_i\eta)\,\bol{s}_\eta
+
\sin\eta\,(\partial_i\varphi)\,\bol{s}_\varphi,\qquad \lr{\cos\varphi,\sin\varphi,0}=\sin\eta\bol{s}+\cos\eta\bol{s}_{\eta}, 
\end{equation}
one computes
\begin{align}
\frac{1}{\rho}\partial_i(\rho\,\partial_i\bol{s})
&=
\left[
\frac{1}{\rho}\nabla\cdot(\rho\nabla\eta)
-\sin\eta\cos\eta\,\abs{\nabla\varphi}^2
\right]\bol{s}_\eta
+
\left[
\frac{1}{\rho\sin\eta}\nabla\cdot\!\lr{\rho\sin^2\eta\,\nabla\varphi}
\right]\bol{s}_\varphi
-\lr{\abs{\nabla\eta}^2+\sin^2\eta\abs{\nabla\varphi}^2} \bol{s}. 
\end{align}
Taking the cross product with $\bol{s}$ 
we obtain
\eq{
\label{eq:s-cross-s}
\frac{\hbar}{2m\rho}\,\bol{s}\times\partial_i(\rho\,\partial_i\bol{s})
=
-\frac{\hbar}{2m\rho\sin\eta}\nabla\cdot\!\lr{\rho\sin^2\eta\,\nabla\varphi}\,\bol{s}_\eta
+
\left[
\frac{\hbar}{2m\rho}\nabla\cdot(\rho\nabla\eta)
-\frac{\hbar}{2m}\sin\eta\cos\eta\,\abs{\nabla\varphi}^2
\right]\bol{s}_\varphi.
}
Substituting \eqref{eq:eta-evolution} and \eqref{eq:phi-evolution} into \eqref{st}, and using   
\eqref{eq:s-cross-B}  and \eqref{eq:s-cross-s}, we conclude that
\eq{
\frac{\partial \bol{s}}{\partial t}
+
(\bol{u}\cdot\nabla)\bol{s}
=
\frac{q}{m}\,\bol{s}\times\bol{B}
+
\frac{\hbar}{2m\rho}\,\bol{s}\times\partial_i(\rho\,\partial_i\bol{s}).
\label{eq:spin-evolution}
}
Finally, with the identity $\p_i\lr{\rho\p_i\bol{s}}=\lr{\nabla\rho\cdot\nabla}\bol{s}+\rho\Delta\bol{s}$, one arrives at the desired evolution equation for the spin field \eqref{eq:hydro-spin}.


\section{Derivation of the Pauli equation for a fluid with spin} \label{sec:PauliDer}
The purpose of this section is to show that a charged fluid endowed with an internal spin degree of freedom naturally satisfies the Pauli equation in its hydrodynamic (Madelung) form, namely system \eqref{eq:hydro-continuity}, \eqref{eq:hydro-momentum}, and \eqref{eq:hydro-spin}.
This derivation generalizes the construction developed in \cite{Sato25QF}, where it was shown that a charged fluid with a constant internal spin distribution \(\bol{s}=\pm\nabla z\) satisfies the Schr\"odinger equation in Madelung form.
Most importantly, the result of the present section shows that the Pauli equation emerges solely from fluid-mechanical and electromagnetic principles, rather than being introduced as an ansatz of the theory.

In a smooth bounded domain \(\Omega\subset\mathbb{R}^3\), we consider a fluid of total charge \(q\), total mass $m$, and with density \(\rho\lr{\bol{x},t}\), normalized so that
\eq{
\int_{\Omega}\rho\,dV=1,
}
where \(dV\) denotes the Euclidean volume element. 
The flow velocity $\bol{u}\lr{\bol{x},t}
=
\bol{u}_T\lr{\bol{x},t}
+
\bol{v}_s\lr{\bol{x},t}$ 
is decomposed into a slow component $\bol{u}_T\lr{\bol{x},t}$, evolving on the time scale $T$, and a fast component $\bol{v}_s\lr{\bol{x},t}$, which characterizes the spin dynamics on a time scale $\tau_s\ll T$. 
By dimensional considerations, for a spin-$\frac12$ particle such as the electron, since $\abs{\bol{v}_s}< c$ 
one expects 
a lower bound 
\eq{
\tau_s > \frac{\hbar}{mc^2}\sim 10^{-21}\,{\rm s}, 
}
 where $c$ is the speed of light. 

Because the fluid is charged, both \(\bol{u}_T=\bol{u}-\bol{v}_s\) and \(\bol{v}_s\) generally generate electromagnetic fields.
The magnetic field induced by \(\bol{u}_T\) is the self-induced macroscopic magnetic field \(\bol{\mc{B}}\lr{\bol{x},t}\), which combines with the external magnetic field \(\bol{B}_{\rm ext}\lr{\bol{x},t}\) to yield the total magnetic field
\(
\bol{B}=\bol{\mc{B}}+\bol{B}_{\rm ext},
\)
playing a role analogous to that of the magnetic field in magnetohydrodynamics (MHD).

If the spin motion \(\bol{v}_s\) is neglected, one recovers the equations of an electromagnetic fluid, namely the Euler--Maxwell system, formulated in terms of \(\lr{\rho,\bol{u},\bol{E},\bol{B}}\): the continuity and momentum equations for the hydrodynamic fields \(\lr{\rho,\bol{u}}\) are coupled to the relevant Maxwell equations for the electromagnetic fields \(\lr{\bol{E},\bol{B}}\), together with an appropriate equation of state for the mechanical pressure.

In the following, we shall assume that the external electromagnetic fields $\lr{\bol{E}_{\rm ext},\bol{B}_{\rm ext}}$ dominate the self-induced ones $\lr{\bol{\mc{E}},\bol{\mc{B}}}$, i.e. $\bol{E}_{\rm ext}\gg\bol{\mc{E}}$ and $\bol{B}_{\rm ext}\gg\bol{\mc{B}}$.  
This enables us to 
simplify the phase space $\lr{\rho,\bol{u},\bol{s},\bol{E},\bol{B}}$ of the system to the triplet $\lr{\rho,\bol{u},\bol{s}}$, and to 
focus on the fluid dynamics  
caused by the spin motion $\bol{v}_s$.  Indeed, 
our interest 
lies in the regime in which \(\bol{v}_s\) cannot be neglected.
In this regime, the magnetic field \(\bol{B}_s\lr{\bol{x},t}\) generated by \(\bol{v}_s\) is governed by the Amp\`ere--Maxwell equation
\eq{
\nabla\cp\bol{B}_s
=
\mu_0 q\rho \bol{v}_s
+
\frac{1}{c^2}\frac{\p\bol{E}}{\p t},
}
where \(\mu_0\) denotes the vacuum permeability and \(c\) the speed of light.
Restricting attention to a non-relativistic regime, we neglect the displacement current; formally, we adopt the limit \(\epsilon_0=1/\lr{\mu_0 c^2}\rightarrow0\).
Equivalently, omitting the term \(c^{-2}\p_t\bol{E}\) amounts to adopting the magnetoquasistatic approximation used in the derivation of MHD, in which the magnetic field is generated predominantly by the material current \(q\rho\bol{v}_s\), while electromagnetic-wave propagation and retardation effects are neglected (see, e.g., \cite{Freidberg2014IdealMHD}).
Then, whenever \(\rho\neq 0\), we may write
\eq{
\bol{v}_s=\frac{\nabla\cp\bol{B}_s}{\mu_0q\rho},\qquad\nabla\cdot\lr{\rho\bol{v}_s}=0.\label{vs}
}
Let
\eq{
\rho\bol{S}=c_s\hbar\rho\bol{s},
}
denote the spin (intrinsic angular momentum) density, where \(c_s\) is a positive dimensionless real constant, and \(\bol{S}\lr{\bol{x},t}=c_s\hbar\bol{s}\lr{\bol{x},t}\) is the spin field.
In the spin-\(\frac12\) particle analogy, \(c_s=1/2\).
The misalignment between \(\bol{s}\) and the magnetic field \(\bol{B}_s\) produces a torque on \(\bol{s}\), and thus determines the evolution of the spin field on the time scale \(\tau_s\) as
\eq{
\rho\lr{\frac{\p}{\p t}+{\bol{v}_s}\cdot\nabla}\bol{s}=\kappa_s\rho\bol{s}\times\frac{q\bol{B}_s}{m},\label{s0}
}
where \(\kappa_s\) is a dimensionless physical constant.
By the assumed scale separation, on the macroscopic time scale \(T\gg\tau_s\), the fast spin dynamics rapidly approaches a local equilibrium self-field \(\bol{B}_s^e\), whose torque vanishes:
\eq{
\bol{s}\times\bol{B}_s^e=\bol{0}.
}
Specifically, it is natural to expect the local equilibrium magnetization to be proportional to the local spin density. We therefore introduce a local constitutive approximation in which \(\bol{B}_s^e\) is proportional to \(\rho\bol{s}\):
\eq{
\bol{B}_s^e=c_{g}\frac{q\mu_0 \hbar}{2m}\rho\bol{s},\qquad c_{g}=g c_s,\label{Bs}
}
where \(g\) is a dimensionless physical constant whose magnitude is approximately \(2\) for spin-\(\frac12\) particles such as the electron. In that case \(c_g=1\). 
We then decompose the self-induced field as
\eq{
\bol{B}_s=\bol{B}_s^e+\delta\bol{B},
}
where \(\delta\bol{B}\) represents the deviation from the local equilibrium \eqref{Bs} generated by the fast motion \(\bol{v}_s\).
In the quasistationary fast-scale regime, equation \eqref{s0} therefore reduces to
\eq{
\rho\lr{\bol{v}_s\cdot\nabla}\bol{s}=\kappa_s\rho\bol{s}\times\frac{q\delta\bol{B}}{m}.\label{s_}
}
On the other hand, on the longer time scale \(T\), the evolution of the spin field \(\bol{s}\) is governed by
\eq{
\rho\lr{\frac{\p}{\p t}+\bol{u}\cdot\nabla}\bol{s}
=\kappa_s\rho\bol{s}\times\frac{q}{m}\lr{\bol{B}+\bol{B}_s^e+\delta\bol{B}}
=\kappa_s\rho\bol{s}\times\frac{q\bol{B}}{m}+\rho\lr{\bol{v}_s\cdot\nabla}\bol{s}. \label{s1}
}
Note that this equation makes it clear that, on the time scale \(T\), the fluid dynamics is governed by the effective velocity field \(\bol{u}_T=\bol{u}-\bol{v}_s\), obtained by subtracting the fast spin motion \(\bol{v}_s\) from the total velocity field \(\bol{u}\), so that $\lr{\p_t+\bol{u}_T\cdot\nabla}\bol{s}=\kappa_s\bol{s}\times q\bol{B}/m$.  
Let us show how equation \eqref{s1} reduces to the spin field equation \eqref{eq:hydro-spin}.
From equations \eqref{vs} and \eqref{Bs} we first observe that, to leading order in $\delta\bol{B}$, 
\eq{
\bol{v}_s=\frac{c_{g}\hbar}{2m}\lr{\nabla\log\rho\cp\bol{s}+\nabla\cp\bol{s}}.\label{vs2}
}
Next, denoting by
\eq{
h_s=\bol{s}\cdot\nabla\cp\bol{s},
}
the helicity density of the spin field, recalling that \(\bol{s}^2=1\), and noting that  
\eq{
\nabla\cdot\bol{B}_s^e=c_g\frac{q\mu_0\hbar}{2m}\nabla\cdot\lr{\rho\bol{s}}=c_g\frac{q\mu_0\hbar}{2m}\rho\lr{\nabla\log\rho\cdot\bol{s}+\nabla\cdot\bol{s}}=0,\label{divBs}
}
implies \(\nabla\cdot\bol{s}=-\nabla\log\rho\cdot\bol{s}\), we obtain the identity
\eq{
\rho\lr{\bol{v}_s\cdot\nabla}\bol{s}=&\frac{c_g\hbar}{2m}\lrs{\lr{\nabla\rho\cp\bol{s}+\rho\nabla\cp\bol{s}}\cdot\nabla}\bol{s}\\
=&\frac{c_g\hbar}{2m}\lrc{
-\lr{\bol{s}\cdot\nabla}\lr{\nabla\rho\cp\bol{s}}
-\lr{\nabla\rho\cp\bol{s}}\cp\lr{\nabla\cp\bol{s}}-\bol{s}\cp\lrs{\nabla\cp\lr{\nabla\rho\cp\bol{s}}}
}\\
&+\frac{c_g\hbar\rho}{2m}
\lrs{\nabla h_s-\lr{\bol{s}\cdot\nabla}\nabla\cp\bol{s}-\bol{s}\cp\lr{\nabla\cp\nabla\cp\bol{s}}}\\
=&\frac{c_g\hbar}{2m}\lrc{
-\nabla\rho\cp\lrs{\lr{\bol{s}\cdot\nabla}\bol{s}}+h_s\nabla\rho-\lr{\nabla\cp\bol{s}\cdot\nabla\rho}\bol{s}-\bol{s}\cp\lrs{\lr{\nabla\cdot\bol{s}}\nabla\rho
-\lr{\nabla\rho\cdot\nabla}\bol{s}}}\\
&+\frac{c_g\hbar\rho}{2m}
\lrc{\nabla h_s-\lr{\bol{s}\cdot\nabla}\nabla\cp\bol{s}+\bol{s}\cp\lrs{\Delta\bol{s}-\nabla\lr{\nabla\cdot\bol{s}}}
}\\
=&\frac{c_g\hbar}{2m}\lrc{
\nabla\lr{\rho h_s}+\rho\bol{s}\cp\lrs{\Delta\bol{s}+\lr{\nabla\log\rho\cdot\nabla}\bol{s}}
}\\
&+\frac{c_g\hbar}{2m}\lrc{
-\lr{\bol{s}\cdot\nabla\rho}\nabla\cp\bol{s}
-\rho\lr{\bol{s}\cdot\nabla}\nabla\cp\bol{s}-\bol{s}\cp\lrs{\lr{\nabla\cdot\bol{s}}\nabla\rho+\rho\nabla\lr{\nabla\cdot\bol{s}}}
}\\
=&\frac{c_g\hbar}{2m}\lrc{
\nabla\lr{\rho h_s}+\rho\bol{s}\cp\lrs{\Delta\bol{s}+\lr{\nabla\log\rho\cdot\nabla}\bol{s}}
}-\frac{c_g\hbar}{2m}\lrc{
\lr{\bol{s}\cdot\nabla}\lr{\rho\nabla\cp\bol{s}}-\bol{s}\times\nabla\lr{\bol{s}\cdot\nabla\rho}
}.\label{vsid}
}
Now assume that the helicity density \(h_s\) and the variations along the spin field \(\bol{s}\cdot\nabla\) are small.
Since a vanishing helicity density \(h_s=0\) corresponds to the Frobenius integrability condition for the spin field \(\bol{s}\), implying that locally \(\bol{s}=\lambda\lr{\bol{x},t}\nabla C\lr{\bol{x},t}\) defines the normal of a surface \(\Sigma_C=\lrc{\bol{x}\in\Omega:C={\rm constant}}\), this is the case, for example, of \(\bol{s}=\pm \nabla z\) and a fluid dynamics developing on the planes \(\Sigma_z=\lrc{\bol{x}\in\Omega:z={\rm constant}}\), so that \(\bol{s}\cdot\nabla=\pm\p/\p z=0\).
Under such smallness assumption, retaining the dominant terms in the last line, and recalling equation \eqref{s1}, leads to the spin field equation
\eq{
\lr{\frac{\p}{\p t}+{\bol{u}}\cdot\nabla}\bol{s}=\kappa_s\bol{s}\times\frac{q\bol{B}}{m}+\frac{c_g\hbar}{2m}\bol{s}\cp\lrs{\Delta\bol{s}+\lr{\nabla\log\rho\cdot\nabla}\bol{s}},\label{eq:spinfluid-spin}
}
which is equivalent to the spin field equation \eqref{eq:hydro-spin} for \(\kappa_s=c_g=1\).
The continuity equation for the fluid on the time scale \(T\),
\eq{
\frac{\p\rho}{\p t}=-\nabla\cdot\lr{\rho\bol{u}},\label{eq:spinfluid-continuity}
}
also exhibits the same form of the continuity equation for the hydrodynamic form of the Pauli equation, equation \eqref{eq:hydro-continuity}. Thus, the remaining step is to show that the fluid also obeys the momentum equation \eqref{eq:hydro-momentum}.
To see this, we will evaluate the total force acting on the fluid by looking at its total energy and at  the amount of work needed to change it.  

We begin by writing 
the total kinetic energy associated with the velocity field $\bol{v}_s$ via \eqref{vs2}, 
\eq{
K_s=&\frac{m}{2}\int_{\Omega}\rho\bol{v}_s^2\,dV\\
=&\frac{c_{g}^2\hbar^2}{8m}\int_{\Omega}\rho\lrs{\abs{\nabla\log\rho}^2-\lr{\nabla\log\rho\cdot\bol{s}}^2+2\nabla\log\rho\cp\bol{s}\cdot\nabla\cp\bol{s}+\lr{\nabla\cp\bol{s}}^2}\,dV.\label{Ks}
}
The first term on the right-hand side
\eq{
\frac{c_g^2\hbar^2}{8m}\int_{\Omega}\rho\abs{\nabla\log\rho}^2\,dV,\label{Eq} 
}
is the ``quantum'' energy responsible  of 
the correction to the classical fluid energy leading to 
the Schr\"odinger equation in Madelung form (see \cite{Sato25QF}). 
The hydrodynamic Pauli system \eqref{eq:hydro-continuity}, \eqref{eq:hydro-momentum}, and \eqref{eq:hydro-spin}, follows by keeping the other terms, which explicitly involve  $\bol{s}$. To see this, we need two identities. First, notice the pointwise identity 
\eq{
\abs{\nabla \bol{s}}^2
=&
\abs{\nabla\times \bol{s}}^2
+
\abs{\nabla\cdot\bol{s}}^2
+
\nabla\cdot\lrs{(\bol{s}\cdot\nabla)\bol{s}-(\nabla\cdot\bol{s})\,\bol{s}}\\
=&\abs{\nabla\cp\bol{s}}^2+\nabla\cdot\lrs{\lr{\bol{s}\cdot\nabla}\bol{s}}-\bol{s}\cdot\nabla\lr{\nabla\cdot\bol{s}}.\label{FrobNorm}
}
Using equations \eqref{divBs} and \eqref{FrobNorm}, noting that $\nabla\log\rho\cdot\bol{s}=-\nabla\cdot\bol{s}$ and  $\Delta\abs{\bol{s}}^2=\Delta 1=0$, 
and performing integration by parts with the boundary conditions 
\eq{
\bol{s}\cdot\bol{n}=
\nabla\cp\bol{s}\cdot\bol{n}=
0\qquad{\rm on}~~\p\Omega, \label{bc}
}
where $\bol{n}$ is the unit outward normal on $\p\Omega$, 
the kinetic energy \eqref{Ks} can be expressed as 
\eq{
K_s
=&\frac{c_g^2\hbar^2}{8m}\int_{\Omega}\rho\lrc{\abs{\nabla\log \rho}^2+\abs{\nabla\bol{s}}^2
-\lr{\nabla\cdot\bol{s}}^2-2\nabla\cdot\lrs{\bol{s}\cp\lr{\nabla\cp\bol{s}}}
-\nabla\cdot\lrs{\lr{\bol{s}\cdot\nabla}\bol{s}}
+\bol{s}\cdot\nabla\lr{\nabla\cdot\bol{s}}
}dV\\
=&\frac{c_g^2\hbar^2}{8m}\int_{\Omega}\rho\lrc{\abs{\nabla\log \rho}^2+\abs{\nabla\bol{s}}^2-\lr{\nabla\cdot\bol{s}}^2
-\nabla\cdot\lrs{\lr{\bol{s}\cdot\nabla}\bol{s}+2\bol{s}\cp\lr{\nabla\cp\bol{s}}}}dV\\
=&\frac{c_g^2\hbar^2}{8m}\int_{\Omega}\rho\lrc{\abs{\nabla\log \rho}^2+\abs{\nabla\bol{s}}^2-\lr{\nabla\cdot\bol{s}}^2
+\nabla\cdot\lrs{\lr{\bol{s}\cdot\nabla}\bol{s}}-\Delta\abs{\bol{s}}^2}dV,\\
=&
\frac{c_g^2\hbar^2}{8m}\int_{\Omega}\rho\lrc{\abs{\nabla\log \rho}^2+\abs{\nabla\bol{s}}^2+
\nabla\cdot\lrs{\lr{\bol{s}\cdot\nabla}\bol{s}-\lr{\nabla\cdot\bol{s}}\bol{s}}
}dV.
}
The second term of the last line of this equation,
\eq{
\frac{c_g^2\hbar^2}{8m}\int_{\Omega}\rho\abs{\nabla\bol{s}}^2\,dV, \label{Eds}
}
is the spin-stress energy term of the 
Takabayasi Hamiltonian (the expectation value of the Pauli Hamiltonian operator \eqref{PauliHop}, see \cite{Tak54}). We will see that it is this term that, once added to the quantum energy \eqref{Eq}, elevates the Schr\"odinger-Madelung picture to the Pauli-Madelung picture. 
In the following we shall therefore neglect the third term
\eq{
\frac{c_g^2\hbar^2}{8m}\int_{\Omega}\rho\nabla\cdot\lrs{\lr{\bol{s}\cdot\nabla}\bol{s}-\lr{\nabla\cdot\bol{s}}\bol{s}}\,dV=\frac{c_g^2\hbar^2}{8m}\int_{\Omega}\rho\nabla\cdot\lrs{
\frac{1}{\rho}\lr{\bol{s}\cdot\nabla}\lr{\rho\bol{s}}
}\,dV.\label{ER}
}
Before proceeding further, it is however useful to clarify the meaning of this approximation: when $\bol{s}$
defines the normal of surface (this is, for example, the case of a constant spin field $\bol{s}=\pm\nabla z$), 
it amounts to 
neglecting the scalar curvature of such surface (see remark below). 
The smallness of this term also  follows by the smallness hypothesis for $\bol{s}\cdot\nabla$ invoked in the derivation of the spin equation. 
\begin{remark}[Interpretation of the energy term \eqref{ER}] 
Suppose that $\bol{s}$ is locally normal to a family of surfaces (the Frobenius integrability condition $\bol{s}\cdot\nabla\cp\bol{s}=0$ holds). Let $R$ denote the scalar curvature of one such surface. If $\{e_1,e_2\}$ is a local orthonormal frame tangent to the surface, then the second fundamental form is given by
\eq{
A_{ab}=-e_a\cdot \nabla_{e_b}\bol{s}.
}
Hence its trace is
\eq{
H=A_{11}+A_{22}=-\nabla\cdot\bol{s},
}
while its squared norm is
\eq{
|A|^2=\sum_{a,b=1}^2 A_{ab}^2
=\sum_{a=1}^2 |\nabla_{e_a}\bol{s}|^2
=
|\nabla\bol{s}|^2-|(\bol{s}\cdot\nabla)\bol{s}|^2.
}
By the Gauss equation for a surface in Euclidean space \cite{doCarmo76}[p. 146],
\eq{
R=H^2-|A|^2,
}
and therefore
\eq{
R=(\nabla\cdot\bol{s})^2-|\nabla\bol{s}|^2+|(\bol{s}\cdot\nabla)\bol{s}|^2.
}
If in addition $\bol{s}$ is hypersurface-orthogonal, then $\bol{s}\cdot(\nabla\times\bol{s})=0$, and for a unit vector field this implies
\eq{
|(\bol{s}\cdot\nabla)\bol{s}|^2=|\nabla\times\bol{s}|^2.
}
Combining this with the identity
\eq{
|\nabla\bol{s}|^2
=
|\nabla\times\bol{s}|^2
+
(\nabla\cdot\bol{s})^2
+
\nabla\cdot\lrs{(\bol{s}\cdot\nabla)\bol{s}-(\nabla\cdot\bol{s})\bol{s}},
}
one finds
\eq{
R=
-\nabla\cdot\lrs{(\bol{s}\cdot\nabla)\bol{s}-(\nabla\cdot\bol{s})\bol{s}}.
}
Thus, in the unit-length hypersurface-orthogonal case, the divergence term appearing above is precisely minus the scalar curvature of the surfaces orthogonal to $\bol{s}$. 
If we assume that $R$ is small 
, then we may neglect the corresponding term in the kinetic energy $K_s$. 
\end{remark}

Keeping the two terms \eqref{Eq} and \eqref{Eds} in the expression of $K_s$, we may now write the total energy of the fluid as
\eq{
\mathscr{H}\lrs{\rho,\bol{u},\bol{s}}=\int_{\Omega}\rho\lrs{
\frac{1}{2}m\bol{u}^2+q\Phi+V+\frac{h\lr{\rho}}{\rho}+\frac{c_g^2\hbar^2}{8m}
\lr{\abs{\nabla\log\rho}^2+\abs{\nabla\bol{s}}^2}-\frac{c_g\kappa_s q\hbar}{2m}\bol{B}\cdot\bol{s}
}\,dV,\label{fE}
}
where the first term is the kinetic energy associated with the flow $\bol{u}$, the second term the potential energy associated with the electrostatic potential $\Phi\lr{\bol{x},t}$, the third term the energy associated with an external potential $V\lr{\bol{x}}$, the fourth term the internal energy associated with a barotropic mechanical pressure field $P\lr{\rho}$ with $dP/d\rho=\rho d^2h/d{\rho^2}$ and $h=h\lr{\rho}$, the fifth and sixth terms the quantum-potential and spin-stress contributions coming from the kinetic energy $K_s$ of the flow $\bol{v}_s$, and the last term the classical interaction energy between the magnetic field $\bol{B}$ and the magnetic moment $c_gq\hbar\bol{s}/2m$ via the coupling constant $\kappa_s$. 

The total force $\bol{F}\lr{\bol{x},t}$ acting on the fluid can now be evaluated by considering the work done on the system. To this end, we consider the infinitesimal displacement of the fluid $\delta\bol{\ell}=\bol{u}\lr{\bol{x},t}\delta t$ caused by the action of the force $\bol{F}$ over an infinitesimal time interval $\delta t$, while enforcing conservation of energy in the case of steady external electromagnetic fields $\bol{E}=-\nabla\Phi\lr{\bol{x}}$ and $\bol{B}\lr{\bol{x}}$:
\eq{
0=\delta\mathscr{H}
=&\int_{\Omega}\lrs{
\frac{1}{2}m\bol{u}^2+q\Phi+V+h'+c_g^2Q+\frac{c_g^2\hbar^2}{8m}\abs{\nabla\bol{s}}^2-\frac{c_g\kappa_s q\hbar}{2m}\bol{B}\cdot\bol{s}
}\delta\rho\,dV\\
&+m\int_{\Omega}\rho\bol{u}\cdot\delta\bol{u}\,dV+\int_{\Omega}\lrs{-\frac{c_g^2\hbar^2}{4m}\lr{\rho\Delta s^k+{\nabla\rho\cdot\nabla s^k}}-\frac{c_g\kappa_s q\hbar}{2m}\rho B^k}\delta s_k\,dV,\label{dH}
}
where $h'=dh/d\rho$ and we used integration by parts with vanishing variations on the boundary $\p\Omega$.

Now observe that $\rho$ is the coefficient of a top form, $\bol{u}$ is a vector field, and $s^k$ is the scalar component of the spin field. Hence their variations are related to $\delta\bol{\ell}$ and $\delta t$ by
\eq{
\delta\rho=-\nabla\cdot\lr{\rho\delta\bol{\ell}},\qquad
\delta\bol{u}=\frac{\bol{F}}{m}\delta t-\lr{\bol{u}\cdot\nabla}\bol{u}\,\delta t,
\qquad
\delta s_k=-\nabla s_k\cdot\delta\bol{\ell}. \label{deltas}
}
Note that magnetic torques have not been included in $\delta s_k$, since they do not contribute to the work done on the system. 
Substituting the expressions \eqref{deltas} for the variations of the fields into equation \eqref{dH}, we obtain
\eq{
\int_{\Omega}\rho\bol{F}\cdot\delta\bol{\ell}\,dV=&
-\int_{\Omega}\rho\nabla\lr{
q\Phi+V+h'+c_g^2Q+\frac{c_g^2\hbar^2}{8m}\abs{\nabla\bol{s}}^2-\frac{c_g\kappa_s q\hbar}{2m}\bol{B}\cdot\bol{s}}\cdot\delta\bol{\ell}\,dV\\
&-\int_{\Omega}\lrc{\frac{c_g^2\hbar^2}{4m}\lrs{\rho\Delta s^k+\lr{\nabla\rho\cdot\nabla s^k}}+\frac{c_g\kappa_s q\hbar}{2m}\rho B^k}\nabla s_k\cdot\delta\bol{\ell}\,dV.
}
Since the displacement $\delta\bol{\ell}$ is arbitrary, this equation determines the force $\bol{F}$ up to a contribution perpendicular to $\delta\bol{\ell}=\bol{u}\delta t$. The only force perpendicular to $\delta\bol{\ell}$ is the magnetic part of the Lorentz force. We thus find
\eq{
\bol{F}=q\bol{u}\cp\bol{B}-\nabla\lr{q\Phi+V+c_g^2Q}-\frac{1}{\rho}\nabla P+\frac{c_g\kappa_s q\hbar}{2m}s_k\nabla B^k-\frac{c_g^2}{\rho}\nabla\cdot\Pi,
}
and hence the momentum equation
\eq{
m\rho\lr{\frac{\p}{\p t}+{\bol{u}\cdot\nabla}}\bol{u}=q\rho\lr{\bol{E}+\bol{u}\cp\bol{B}}-\rho\nabla\lr{V+c_g^2Q}-\nabla P+\frac{c_g\kappa_s q\hbar}{2m}\rho s_k\nabla B^k-
c^2_g\nabla\cdot\Pi.\label{eq:spinfluid-momentum}
}
We note that the momentum equation \eqref{eq:spinfluid-momentum} remains valid in the case of time-dependent external electromagnetic fields, because the Lorentz force retains the same form.
Equation \eqref{eq:spinfluid-momentum} reduces to the momentum equation \eqref{eq:hydro-momentum} of the hydrodynamic Pauli equation for $\kappa_s=c_g=1$, $P={\rm constant}$, and $V={\rm constant}$.

We may summarize the derivation of the present section as follows:

\begin{proposition}[Fluid with spin ordering leading to the hydrodynamic form of the Pauli equation] \label{prop:ordering}
Let $\epsilon>0$ denote a small ordering parameter, and assume that
\eq{
\frac{\tau_s}{T}\sim h_s\sim \bol{s}\cdot\nabla\sim \frac{\abs{\bol{\mc{B}}}}{\abs{\bol{B}_{\rm ext}}}\sim\frac{\abs{\bol{\mc{E}}}}{\abs{\bol{E}_{\rm ext}}}\sim P\sim\epsilon,\qquad \kappa_s=c_g=1.
}
Then, to leading order, a fluid with spin obeys the hydrodynamic (Madelung) form of the Pauli equation, eqs. \eqref{eq:hydro-continuity}, \eqref{eq:hydro-momentum}, and \eqref{eq:hydro-spin}.
\end{proposition}

\section{Hamiltonian structure of a fluid with spin}\label{sec:Hamiltonian}
The aim of the present section is to discuss the Hamiltonian structure of the hydrodynamic (Madelung) form of the Pauli equation, namely
\eqref{eq:hydro-continuity},
\eqref{eq:hydro-momentum},
and \eqref{eq:hydro-spin},
or, equivalently, of the governing equations
\eqref{eq:spinfluid-continuity},
\eqref{eq:spinfluid-momentum},
and \eqref{eq:spinfluid-spin}
for a fluid with spin. 
For the classical foundations of the noncanonical Hamiltonian structure of fluids and plasmas, we refer the reader to the seminal works of Morrison \cite{Morrison1982} and Littlejohn \cite{Littlejohn1982}. 


We begin by observing that the phase space of the system is spanned by the fields
\(
\bol{z}=\lr{\rho,\bol{u},\bol{s}}
\), with $\rho>0$ and $\abs{s}=1$. 
Moreover, for $\kappa_s=c_g=1$,  $V=h=0$, 
and $\bol{u}\cdot\bol{n}=\nabla\rho\cdot\bol{n}=0$ on $\p\Omega$, 
the total energy (Hamiltonian) \eqref{fE} coincides with the expectation value of the Pauli Hamiltonian operator \eqref{PauliHop}, namely
\eq{
\mathscr{H}\lrs{\rho,\bol{u},\bol{s}}=\int_{\Omega}\Psi^\dagger H\Psi\,dV.\label{HtoHop}
}

The Hamiltonian structure of the system is then encoded by the Hamiltonian \eqref{fE} together with the Poisson operator
\begin{equation}
\mc{J}=
\begin{bmatrix}
0 & -\frac{1}{m}\nabla\cdot & 0 & 0 & 0\\
-\frac{1}{m}\nabla & -\frac{1}{m\rho}\lrs{\nabla\cp\lr{\bol{u}+\frac{q\bol{A}}{m}}}\cp & \frac{1}{m\rho}\nabla s_1 & \frac{1}{m\rho}\nabla s_2 & \frac{1}{m\rho}\nabla s_3\\
0 & -\frac{1}{m\rho}\nabla s_1\cdot & &  &\\
0 & -\frac{1}{m\rho}\nabla s_2\cdot & & -\frac{2}{c_g\hbar\rho}\bol{s}\cp &\\
0 & -\frac{1}{m\rho}\nabla s_3\cdot & & &
\end{bmatrix},\label{J}
\end{equation}
which defines the noncanonical Poisson bracket
\eq{
\lrc{F,G}=&
\int_{\Omega}
\frac{\delta F}{\delta\bol{z}}\cdot\mc{J}\frac{\delta G}{\delta\bol{z}}
\,dV\\
=&
\int_{\Omega}
\lrc{-\frac{1}{m}\frac{\delta F}{\delta\rho}\nabla\cdot\frac{\delta G}{\delta \bol{u}}
-\frac{1}{m}\frac{\delta F}{\delta \bol{u}}\cdot\nabla\lr{\frac{\delta G}{\delta\rho}}
-\frac{1}{m\rho}\frac{\delta F}{\delta\bol{u}}\cdot\lrs{\nabla\cp\lr{\bol{u}+\frac{q\bol{A}}{m}}}\cp\frac{\delta G}{\delta\bol{u}}}\,dV\\
&+\int_{\Omega}\lrs{
\frac{1}{m\rho}\lr{\frac{\delta G}{\delta s_k}\frac{\delta F}{\delta \bol{u}}-\frac{\delta F}{\delta s_k}\frac{\delta G}{\delta\bol{u}}}\cdot\nabla s_k
-\frac{2}{c_g\hbar\rho}\frac{\delta F}{\delta \bol{s}}\cdot\bol{s}\cp\frac{\delta G}{\delta\bol{s}}
}
\,dV,\label{PB}
}
for smooth functionals $F$ and $G$ of the phase space variables $\bol{z}$. 

It is worth observing that the Poisson bracket \eqref{PB} reduces, upon dropping the spin field $\bol{s}$ and the magnetic field $\bol{B}$, to the classical noncanonical Poisson bracket of an Euler fluid; see, for example, \cite{Morrison1998}, \cite[Part I, chap. 6]{ArnoldKhesin2021}, and \cite{Hamdi2015}. 

We have the following.

\begin{proposition}[Hamiltonian structure]\label{prop:Hamiltonian-structure}
The hydrodynamic (Madelung) form of the Pauli equation, \eqref{eq:hydro-continuity}, \eqref{eq:hydro-momentum}, and \eqref{eq:hydro-spin}, as well as the governing equations for a fluid with spin, \eqref{eq:spinfluid-continuity}, \eqref{eq:spinfluid-momentum}, and \eqref{eq:spinfluid-spin}, can be written in the noncanonical Hamiltonian form
\eq{
\frac{\p\bol{z}}{\p t}=\lrc{\bol{z},\mathscr{H}}.\label{HEoM}
}
\end{proposition}

\noindent Formally, proposition \ref{prop:Hamiltonian-structure} follows by verifying that the bracket \eqref{PB} satisfies the Poisson bracket axioms (bilinearity, antisymmetry, the Leibniz rule, and the Jacobi identity) under appropriate boundary conditions, 
by computing 
the functional derivatives of the Hamiltonian $\mathscr{H}$ 
and then substituting 
their 
expressions into \eqref{HEoM}.
The detailed proof is given below. 

\begin{proof}
The only nontrivial point is the Jacobi identity for the bracket \eqref{PB}.  
To establish it, we introduce the variables
\eq{
\bol{M}=\rho\lr{m\bol{u}+q\bol{A}},\qquad
\boldsymbol{\Sigma}=\frac{c_g\hbar}{2}\rho\bol{s}.
\label{MSvars}
}
Let \(T:(\rho,\bol{u},\bol{s})\mapsto(\rho,\bol{M},\bol{\Sigma})\) denote the change of variables \eqref{MSvars}. 
Since \(\rho>0\) and \(\bol{A}\) is an externally prescribed field, the map
\(
(\rho,\bol{u},\bol{s})\mapsto(\rho,\bol{M},\boldsymbol{\Sigma})
\)
is invertible on the phase space. The inverse transformation $T^{-1}$ is given by 
\eq{
\bol{u}=\frac{\bol{M}}{m\rho}-\frac{q\bol{A}}{m},\qquad
\bol{s}=\frac{2}{c_g\hbar}\frac{\boldsymbol{\Sigma}}{\rho}.
} 
For any functional \(F[\rho,\bol{u},\bol{s}]\), define the corresponding functional in the new variables by
$\overline F[\rho,\bol{M},\bol{\Sigma}]=(F\circ T^{-1})[\rho,\bol{M},\bol{\Sigma}]$. 
A straightforward application of the chain rule yields
\sys{
\frac{\delta F}{\delta \bol{u}}
=&m\rho\,\frac{\delta \overline{F}}{\delta \bol{M}},\\[0.2em]
\frac{\delta F}{\delta \bol{s}}
=&\frac{c_g\hbar}{2}\rho\,\frac{\delta \overline{F}}{\delta \boldsymbol{\Sigma}},\\[0.2em]
\frac{\delta F}{\delta \rho}
=&\frac{\delta \overline{F}}{\delta \rho}
+\frac{\bol{M}}{\rho}\cdot\frac{\delta \overline{F}}{\delta \bol{M}}
+\frac{\boldsymbol{\Sigma}}{\rho}\cdot\frac{\delta \overline{F}}{\delta \boldsymbol{\Sigma}}.
}{chain}
Substituting \eqref{chain} into \eqref{PB}, 
dropping boundary integrals, 
and simplifying, one obtains
\eq{
\lrc{\overline{F},\overline{G}}
=&
-\int_{\Omega}
\bol{M}\cdot
\lrs{
\lr{\frac{\delta \overline{F}}{\delta \bol{M}}\cdot\nabla}\frac{\delta \overline{G}}{\delta \bol{M}}
-
\lr{\frac{\delta \overline{G}}{\delta \bol{M}}\cdot\nabla}\frac{\delta \overline{F}}{\delta \bol{M}}
}
\,dV\\
&-\int_{\Omega}
\rho\lrs{
\frac{\delta \overline{F}}{\delta \bol{M}}\cdot\nabla\frac{\delta \overline{G}}{\delta \rho}
-
\frac{\delta \overline{G}}{\delta \bol{M}}\cdot\nabla\frac{\delta \overline{F}}{\delta \rho}
}
\,dV\\
&-\int_{\Omega}
\Sigma_k\lrs{
\frac{\delta \overline{F}}{\delta \bol{M}}\cdot\nabla\frac{\delta \overline{G}}{\delta \Sigma_k}
-
\frac{\delta \overline{G}}{\delta \bol{M}}\cdot\nabla\frac{\delta \overline{F}}{\delta \Sigma_k}
}
\,dV\\
&+\int_{\Omega}
\boldsymbol{\Sigma}\cdot
\frac{\delta \overline{F}}{\delta \boldsymbol{\Sigma}}
\cp
\frac{\delta \overline{G}}{\delta \boldsymbol{\Sigma}}
\,dV.
\label{PBLP}
}
Here the first line is the standard Lie--Poisson bracket for compressible Euler flow written in momentum-density variables, the second and third lines are the semidirect-product terms describing the advection of the density \(\rho\) and of the spin density \(\boldsymbol{\Sigma}\), and the last line is the pointwise Lie--Poisson bracket on \(\mathfrak{so}(3)^*\).

Therefore \eqref{PBLP} is the standard Lie--Poisson bracket on the dual of the semidirect-product Lie algebra
\eq{
\mathfrak{g}
=
\mathfrak{X}(\Omega)\ltimes\lr{C^\infty(\Omega)\oplus C^\infty\!\lr{\Omega,\mathfrak{so}(3)^*}},
}
where \(\mathfrak{X}(\Omega)\) denotes the Lie algebra of smooth vector fields on \(\Omega\), endowed with the usual commutator bracket, \(C^\infty(\Omega)\) denotes the space of smooth scalar fields on \(\Omega\), and \(C^\infty\!\lr{\Omega,\mathfrak{so}(3)^*}\) denotes the space of smooth maps from \(\Omega\) to \(\mathfrak{so}(3)^*\). Here \(\mathfrak{X}(\Omega)\) acts on \(\rho\) and \(\boldsymbol{\Sigma}\) by Lie derivative, while \(C^\infty\!\lr{\Omega,\mathfrak{so}(3)^*}\) carries the pointwise coadjoint bracket. Hence \eqref{PBLP} satisfies bilinearity, antisymmetry, the Leibniz rule, and the Jacobi identity. Since \eqref{PB} is obtained from \eqref{PBLP} by the invertible change of variables \eqref{MSvars}, it follows that \eqref{PB} is also a Poisson bracket. 

It remains to verify that the Hamiltonian evolution generated by \eqref{PB} and \(\mathscr{H}\) reproduces the equations of motion.
Using the functional derivatives
\sys{
\frac{\delta\mathscr{H}}{\delta\rho}=&\frac{1}{2}m\bol{u}^2+q\Phi+V+h'+c_g^2Q+\frac{c_g^2\hbar^2}{8m}\abs{\nabla\bol{s}}^2-\frac{c_g\kappa_s q\hbar}{2m}\bol{B}\cdot\bol{s},\\
\frac{\delta\mathscr{H}}{\delta\bol{u}}=&m\rho\bol{u},\\
\frac{\delta\mathscr{H}}{\delta{s}_k}=&-\frac{c_g^2\hbar^2}{4m}\lr{\nabla\rho\cdot\nabla s^k+\rho\Delta s^k}-\frac{c_g\kappa_s q\hbar}{2m}\rho B^k,
}{varsproof}
and substituting them into
\eq{
\frac{\p\bol{z}}{\p t}=\lrc{\bol{z},\mathscr{H}},
}
one recovers, by direct computation, the continuity equation \eqref{eq:spinfluid-continuity}, the momentum equation \eqref{eq:spinfluid-momentum}, and the spin equation \eqref{eq:spinfluid-spin}. In the special case \(\kappa_s=c_g=1\), these reduce to \eqref{eq:hydro-continuity}, \eqref{eq:hydro-momentum}, and \eqref{eq:hydro-spin}. This proves the proposition.
\end{proof} 

\section{Many-particle system as a higher-dimensional fluid with spin}\label{sec:NPauli} 

The aim of the present section is to show that a system of $N$ particles obeying the Pauli equation \eqref{Pauli}, or equivalently, $N$ fluids with spin, each obeying the system of equations \eqref{eq:spinfluid-continuity}, \eqref{eq:spinfluid-momentum}, and \eqref{eq:spinfluid-spin}, and interacting via some binary interaction potential energy
\eq{
V_{ij}\lr{\bol{z}_i,\bol{z}_j,\bol{x}_i,\bol{x}_j},\qquad i,j=1,...,N,\label{Vij}
}
where $\bol{z}_i$ denotes the hydrodynamic state variables of species $i$, can be described as a single higher-dimensional fluid flow.  

We begin by defining the state variables of the $N$-fluid system. 
Each fluid occupies a smooth bounded domain $\Omega_i\subset\mathbb{R}^3$, $i=1,...,N$. 
Let $\rho_i\lr{\bol{x}_i,t}$, $\bol{u}_i\lr{\bol{x}_i,t}$, and $\bol{s}_i\lr{\bol{x}_i,t}$ denote the density, fluid velocity, and spin field of fluid $i$, respectively, with $i=1,...,N$ and $\bol{x}_i\in\Omega_i$. We set
\eq{
\bol{z}_i=\lr{\rho_i,\bol{u}_i,\bol{s}_i}.  
}
We denote by
\eq{
\bol{X}=\lr{\bol{x}_1,\dots,\bol{x}_N}\in W,\qquad d\bol{X}=\prod_{i=1}^NdV_i=dV_1\dots dV_N,
}
a point in the domain $W=\Omega_1\times\dots\times\Omega_N\subset\mathbb{R}^{3N}$ and the Euclidean volume element on $W$, respectively. 
The density of the $N$-fluid system is given by
\eq{
\varrho\lr{\bol{X},t}=\prod_{i=1}^N\rho_i\lr{\bol{x}_i,t}=\rho_1\lr{\bol{x}_1,t}\times \dots \times \rho_N\lr{\bol{x}_N,t}.\label{Nrho}
}
The density \eqref{Nrho} satisfies the normalization
\eq{
\int_W\varrho\,d\bol{X}=1. 
}
Let $\mc{X}\lr{W}$ denote the set of vector fields on $W$. 
The velocity field of the $N$-fluid system is defined as the $3N$-dimensional vector field $\bol{U}\lr{\bol{X},t}\in\mc{X}\lr{W}$ with expression 
\eq{
\bol{U}=\sum_{i=1}^N\bol{u}_i\lr{\bol{x}_i,t}=\sum_{i=1}^N\sum_{j=1}^3u_i^{j}\p_{j}^i, \label{NU}
}
where $\p_j^i$ denotes the $j$th tangent vector $\p/\p x_i^j$, $j=1,2,3$, on $\Omega_i$. 
The spin field $\bol{\mf{s}}\lr{\bol{X},t}\in \mc{X}\lr{W}$ of the $N$-fluid system is defined in a similar manner,
\eq{
\bol{\mf{s}}=\sum_{i=1}^N\bol{s}_i\lr{\bol{x}_i,t}=\sum_{i=1}^N\sum_{j=1}^3 s_i^j\p_j^i.\label{Ns}
}
Finally, we define the state variable
\eq{
\bol{Z}=\lr{\varrho,\bol{U},\bol{\mf{s}}}. 
}

Each fluid obeys the system of equations \eqref{eq:spinfluid-continuity}, \eqref{eq:spinfluid-momentum}, and \eqref{eq:spinfluid-spin}, augmented by the pairwise interaction \eqref{Vij}, which, to simplify the analysis, we take to be of the form of a spin-spin interaction 
\eq{
V_{ij}\lr{\bol{x}_i,\bol{x}_j,\bol{s}_i,\bol{s}_j},\qquad V_{ij}=V_{ji},\qquad i,j=1,...,N.
}
We note that the derivation would not change for the  general interaction \eqref{Vij}, but it would   lead to more complicated interaction terms in the governing equations.  
Introducing the mean interaction potential energy acting on a fluid parcel of fluid $i$,
\eq{
\Phi_i\lr{\bol{x}_i,t}=\sum_{j=1}^N\int_{\Omega_j}\rho_j V_{ij}\,dV_j, 
}
the governing equations for fluid $i$ are 
\sys{ 
\frac{\p\rho_i}{\p t}=&-\nabla_i\cdot\lr{\rho_i\bol{u}_i},\label{N-continuity}\\
m\rho_i\lr{\frac{\p}{\p t}+\bol{u}_i\cdot\nabla_i}\bol{u}_i=&
q\rho_i\lr{\bol{E}_i+\bol{u}_i\cp\bol{B}_i}-\rho_i\nabla_i\lr{V_i+\Phi_i+c_g^2Q_i}\\
&-\nabla_i P_i+\frac{c_g\kappa_sq\hbar}{2m}\rho_is_{ik}\nabla_i B_i^k-c_g^2\nabla_i\cdot \Pi_i,\\
\lr{\frac{\p}{\p t}+\bol{u}_i\cdot\nabla_i}\bol{s}_i=&\kappa_s\bol{s}_i\times\frac{q\bol{B}_i}{m}+\frac{c_g\hbar}{2m}\bol{s}_i\times
\lrs{\Delta_i\bol{s}_i+\lr{\nabla_i\log\rho_i\cdot\nabla_i}\bol{s}_i}-\frac{2}{c_g\hbar}\bol{s}_i\cp \frac{\p\Phi_i}{\p\bol{s}_i}.
}{ispinfluid}
Here a lower index $i$ denotes quantities, differential operators, and evaluations corresponding to the $i$th fluid. For example, $\bol{E}_i=\bol{E}\lr{\bol{x}_i,t}$ and $B_i^k$ denotes the $k$th component of $\bol{B}_i$.  

Notice also that the $N$-fluid Hamiltonian is given by 
\eq{
\mathscr{H}_{\rm tot}\lrs{\bol{Z}}=\sum_{i=1}^N\mathscr{H}_i+\frac{1}{2}\sum_{i,j=1}^N\int_{\Omega_i\times\Omega_j}\rho_i\rho_jV_{ij}\,dV_idV_j
=\sum_{i=1}^N\lr{\mathscr{H}_i+\frac{1}{2}\int_{\Omega_i}\rho_i\Phi_i\,dV_i},
}
where $\mathscr{H}_i$ is the Hamiltonian of fluid $i$, and that system \eqref{ispinfluid} follows by substituting $\mathscr{H}_{\rm tot}$ into the Poisson bracket \eqref{PB}.  

It is convenient to define the total force $\bol{F}_i$ and the total torque $\bol{\tau}_i$ acting on a parcel of the $i$th fluid by
\eq{
\bol{F}_i=q\lr{\bol{E}_i+\bol{u}_i\cp\bol{B}_i}-\nabla_i\lr{V_i+\Phi_i+c_g^2Q_i}
-\frac{1}{\rho_i}\nabla_i P_i+\frac{c_g\kappa_sq\hbar}{2m}s_{ik}\nabla_i B_i^k-\frac{c_g^2}{\rho_i}\nabla_i\cdot \Pi_i,
}
and 
\eq{
\bol{\tau}_i=\kappa_s\bol{s}_i\times\frac{q\bol{B}_i}{m}+\frac{c_g\hbar}{2m}\bol{s}_i\times
\lrs{\Delta_i\bol{s}_i+\lr{\nabla_i\log\rho_i\cdot\nabla_i}\bol{s}_i}-\frac{2}{c_g\hbar}\bol{s}_i\cp \frac{\p\Phi_i}{\p\bol{s}_i},
}
so that system \eqref{ispinfluid} can be written as
\sys{
\frac{\p\rho_i}{\p t}+\nabla_i\cdot\lr{\rho_i\bol{u}_i}=&0,\\
\lr{\frac{\p}{\p t}+\bol{u}_i\cdot\nabla_i}\bol{u}_i=&\frac{\bol{F}_i}{m},\label{mom}\\
\lr{\frac{\p}{\p t}+\bol{u}_i\cdot\nabla_i}\bol{s}_i=&\bol{\tau}_i.\label{spi}
}{ispinfluid2}
Let us introduce the $3N$-dimensional gradient operator 
\eq{
\nabla=\sum_{i=1}^N\nabla_i. 
}
We have the following. 
\begin{proposition}[Madelung plasma]\label{prop:MPlasma}
Under tangential boundary conditions $\bol{u}_i\cdot\bol{n}_i=0$ on $\p\Omega_i$, 
where $\bol{n}_i$ denotes the unit outward normal on $\p\Omega_i$, 
system \eqref{ispinfluid2} is equivalent to the $3N$-dimensional fluid with spin system 
\sys{
\frac{\p\varrho}{\p t}+\nabla\cdot\lr{\varrho\bol{U}}=&0\label{Ncon},\\
\lr{\frac{\p}{\p t}+\bol{U}\cdot\nabla}\bol{U}=&\frac{\bol{{F}}}{m}\label{Nmom},\\
\lr{\frac{\p}{\p t}+\bol{U}\cdot\nabla}\bol{\mf{s}}=&\bol{\tau}\label{Nspi},
}
{MPlasma} 
where
\eq{
\bol{{F}}=\sum_{i=1}^N\bol{F}_i,\qquad  \bol{\tau}=\sum_{i=1}^N\bol{\tau}_i.
}
We call such a system a Madelung plasma.  
\end{proposition}

\begin{proof}
The equivalence of the momentum and spin field equations \eqref{Nmom} and \eqref{Nspi} with \eqref{mom} 
and \eqref{spi}
follows by expressing $\bol{U}$ and $\mathfrak{s}$ in terms of the original components $\bol{u}_i$ and $\bol{s}_i$. The $i$th continuity equation can be obtained by integrating \eqref{Ncon} on $\Omega_1\times...\times\Omega_{i-1}\times\Omega_{i+1}\times...\times\Omega_N$ and by applying boundary conditions. 
\end{proof}

\begin{remark}[A fluid quantum computer]\label{rem:quantum-computing} 
Due to the quantum-fluid correspondence linking the Pauli equation \eqref{Pauli} to the fluid with spin system \eqref{eq:spinfluid-continuity}, \eqref{eq:spinfluid-momentum}, and \eqref{eq:spinfluid-spin}, Proposition \ref{prop:MPlasma} implies that a system of $N$ Pauli particles can be encoded into the $3N$-dimensional fluid with spin system \eqref{MPlasma}. In principle, this suggests that real or virtual fluids can be used to perform quantum simulations, and conversely that quantum computers can be used to perform fluid simulations. 
More generally, 
the interference properties inherent to fluid systems may provide an alternative framework to implement quantum-type algorithms. 
\end{remark}


\section*{Statements and Declarations}
\subsection*{Data availability}
Data sharing not applicable to this article as no datasets were generated or analysed during the current study.


\subsection*{Competing interests} 
The authors have no competing interests to declare that are relevant to the content of this article.

\subsection*
{Disclaimer}
This is a preprint. It has not been peer-reviewed or submitted for publication. Please contact the author before using.

\printbibliography



 



\end{document}